\documentclass[11pt]{article}
\usepackage{amsmath,amssymb}
\usepackage{graphicx}
\usepackage{amsfonts}

\usepackage{amsthm}
\newtheorem{thm}{Theorem}[section]

\newtheorem{lemma}[thm]{Lemma}

\usepackage[USenglish]{babel} %francais, polish, spanish, ...
\usepackage[T1]{fontenc}
\usepackage[ansinew]{inputenc}
\usepackage{mathrsfs}
\usepackage{stmaryrd} % for \boxbar
\usepackage{epstopdf} % for adding eps figures
\usepackage{verbatim}
\usepackage{MnSymbol}  %for << and >>, that is \llangle \rrangle
\usepackage{accents}   %for making \dddot look nice
\usepackage{bbm}			 % for \mathbbm{}; makes 1 the mathbb style
\usepackage{amsbsy}		 % for making bold Greek letters
\usepackage{color}     % for changing the font color
\usepackage{authblk}   % for multiple affiliations of the authors

\numberwithin{equation}{section}
\numberwithin{figure}{section}

\title{Stochastic variational principles for the collisional Vlasov-Maxwell and Vlasov-Poisson equations}
\date{}
\author[1,2]{Tomasz M. Tyranowski\thanks{\texttt{tomasz.tyranowski@ipp.mpg.de}}}

\affil[1]{\small Max-Planck-Institut f\"ur Plasmaphysik \authorcr Boltzmannstra{\ss}e 2, 85748 Garching, Germany}
\affil[2]{\small Technische Universit\"{a}t M\"{u}nchen, Zentrum Mathematik \authorcr Boltzmannstra{\ss}e 3, 85748 Garching, Germany}

\begin{document}

\maketitle

\begin{abstract}
In this work we recast the collisional Vlasov-Maxwell and Vlasov-Poisson equations as systems of coupled stochastic and partial differential equations, and we derive stochastic variational principles which underlie such reformulations. We also propose a stochastic particle method for the collisional Vlasov-Maxwell equations and provide a variational characterization of it, which can be used as a basis for a further development of stochastic structure-preserving particle-in-cell integrators.
\end{abstract}

%%%%%%%%%%%%%%%%%%%%%%%%%%%%%%%%%%%%%%%%%%%%%%%%%%%%%%%%%%%%%%%%%%%%%%%%%%%%%%%%%%%%
%  INTRODUCTION
%%%%%%%%%%%%%%%%%%%%%%%%%%%%%%%%%%%%%%%%%%%%%%%%%%%%%%%%%%%%%%%%%%%%%%%%%%%%%%%%%%%%
\section{Introduction}
\label{sec:intro}

The collisional Vlasov equation

\begin{equation}
\label{eq: Collisional Vlasov equation}
\frac{\partial f}{\partial t} + \mathbf{v}\cdot\nabla_x f + \frac{q}{m}(\mathbf{E}+\mathbf{v} \times \mathbf{B})\cdot \nabla_v f = C[f]
\end{equation}

\noindent
describes the time evolution of the particle density function $f=f(\mathbf{x}, \mathbf{v}, t)$ of plasma consisting of charged particles of charge $q$ and mass $m$ which undergo collisions described by the collision operator $C[f]$, and are subject to the electric $\mathbf{E}=\mathbf{E}(\mathbf{x},t)$ and magnetic $\mathbf{B}=\mathbf{B}(\mathbf{x},t)$ fields. The vectors $\mathbf{x}=(x^1,x^2,x^3)$ and $\mathbf{v}=(v^1,v^2,v^3)$ denote positions and velocities, respectively. For simplicity, we restrict ourselves to one-spiece plasmas. Usually, the particle density function is normalized, so that the total number of particles is $N_{tot}=\iint f(\mathbf{x},\mathbf{v},t)d^3\mathbf{v}d^3\mathbf{x}$. However, in this work we would like to treat $f$ as a probability density function, and therefore we will use the normalization $\iint f(\mathbf{x},\mathbf{v},t)d^3\mathbf{v}d^3\mathbf{x}=1$ instead. A self-consistent model of plasma is obtained by coupling \eqref{eq: Collisional Vlasov equation} with the Maxwell equations

\begin{subequations}
\label{eq: Maxwell equations}
\begin{align}
\label{eq: Maxwell equations 1}
\nabla_x \cdot \mathbf{E}&=\rho, \\
\label{eq: Maxwell equations 2}
\nabla_x \cdot \mathbf{B}&=0, \\
\label{eq: Maxwell equations 3}
\nabla_x \times \mathbf{E}&= -\frac{\partial \mathbf{B}}{\partial t}, \\
\label{eq: Maxwell equations 4}
\nabla_x \times \mathbf{B}&= \frac{\partial \mathbf{E}}{\partial t} + \mathbf{J},
\end{align}
\end{subequations}

\noindent
where

\begin{equation}
\label{eq: Charge and current densities}
\rho(\mathbf{x},t) = q N_{tot}\int_{\mathbb{R}^3} f(\mathbf{x}, \mathbf{v}, t)\,d^3\mathbf{v}, \qquad\qquad \mathbf{J}(\mathbf{x},t) = q N_{tot} \int_{\mathbb{R}^3} \mathbf{v} f(\mathbf{x}, \mathbf{v}, t)\,d^3\mathbf{v},
\end{equation}

\noindent
denote the charge density and the electric current density, respectively, and the factor $N_{tot}$ is due to our normalization. The system \eqref{eq: Collisional Vlasov equation}-\eqref{eq: Charge and current densities} is usually referred to as the Vlasov-Maxwell equations. It will also be convenient to express the electric and magnetic fields in terms of the scalar $\varphi(\mathbf{x},t)$ and vector $\mathbf{A}(\mathbf{x},t)$ potentials

\begin{subequations}
\label{eq: Scalar and vector potentials}
\begin{align}
\label{eq: Scalar and vector potentials 1}
\mathbf{E}&= -\nabla_x \varphi - \frac{\partial \mathbf{A}}{\partial t},\\
\mathbf{B}&= \nabla_x \times \mathbf{A},
\end{align}
\end{subequations}

\noindent
as is typical in electrodynamics. The Vlasov-Poisson equations are an approximation of the Vlasov-Maxwell equations in the nonrelativistic zero-magnetic field limit (see Section~\ref{sec: Variational principle for the Vlasov-Poisson equations}). The main goal of this work is to provide a variational characterization of the Vlasov-Maxwell and Vlasov-Poisson equations via a stochastic Lagrange-d'Alembert type of a principle.

Variational principles have proved extremely useful in the study of nonlinear evolution partial differential equations (PDEs). For instance, they often provide physical insights into the problem being considered; facilitate discovery of conserved quantities by relating them to symmetries via Noether's theorem; allow one to determine approximate solutions to PDEs by minimizing the action functional over a class of test functions (see, e.g., \cite{Cooper1993}); and provide a way to construct a class of numerical methods called variational integrators (see \cite{MarsdenPatrickShkoller}, \cite{MarsdenWestVarInt}). A variational principle for the collisionless Vlasov-Maxwell equations was first proposed in \cite{Low1958Lagrangian}. It has been used to derive various particle discretizations of the Vlasov-Maxwell and Vlasov-Poisson equations (see \cite{Evstatiev2013}, \cite{Lewis1970}, \cite{Lewis1972}, \cite{Shadwick2014}, \cite{Stamm2014}), including structure-preserving variational particle-in-cell methods (\cite{Squire2012}, \cite{XiaoLiuQin2013}, \cite{XiaoQinLiu2018}). It has also been applied to gyrokinetic theory (see, e.g., \cite{BottinoSonnendrucker2015}, \cite{SugamaGyrokinetic2000}). For other formulations and extensions see also \cite{YeMorrisonActionPrinciples}.

A structure-preserving description of collisional effects is far less developed. A metriplectic framework for the Vlasov-Maxwell-Landau equations has been presented in \cite{HirvijokiKrausMetriplectic} and \cite{KrausHirvijoki2017}. More recently, a stochastic variational principle has been proposed in \cite{KrausTyranowski2019} to describe collisional effects for the Vlasov equation with a fixed external electric field. To the best of our knowledge, to date no variational principle has been derived for the collisional Vlasov-Maxwell and Vlasov-Poisson equations. In this work we extend the notion of the stochastic Lagrange-d'Alembert principle presented in \cite{KrausTyranowski2019} to plasmas evolving in self-consistent electromagnetic fields. The main idea of our approach is to interpret the Vlasov equation \eqref{eq: Collisional Vlasov equation} as a Fokker-Planck equation and consider the associated stochastic differential equations.

The idea of using stochastic differential equations to model collisions has been pursued by a number of authors over the last few decades; see, e.g., \cite{AllenVictory1994}, \cite{BobylevNanbu2000}, \cite{CadjanIvanov1999}, \cite{CohenDimits2010}, \cite{Dimits2013}, \cite{Frank2003}, \cite{FuQin2020}, \cite{HavlakVictory1996}, \cite{Kleiber2011}, \cite{KrausTyranowski2019}, \cite{Lemons2009}, \cite{Manheimer1997}, \cite{NeunzertVlasovFokkerPlanck1984}, \cite{Sherlock2008}, \cite{Sonnendrucker2015}, \cite{ZhangQin2020}. 

There has been an ever growing body of literature dedicated to stochastic variational principles in recent years. Stochastic variational principles allow the introduction of noise into systems in such a way that the resulting probabilistic models retain all or some of the geometric properties of their deterministic counterparts. For this reason stochastic variational principles have been considered in the context of Lagrangian and Hamiltonian mechanics (\cite{Bismut}, \cite{BouRabeeConstrainedSVI}, \cite{BouRabeeSVI}, \cite{BouRabeeOwhadi2010}, \cite{HolmTyranowskiGalerkin}, \cite{KrausTyranowski2019}, \cite{LaCa-Or2008}, \cite{WangPHD}), soliton dynamics (\cite{HolmTyranowskiSolitons}, \cite{HolmTyranowskiVirasoro}), fluid dynamics (\cite{ArnaudonCruzeiro2014}, \cite{ChenCruzeiroRatiu2018}, \cite{CotterHolm2017}, \cite{CrisanHolm2017}, \cite{Cruzeiro2020}, \cite{GayBalmazHolm2017}, \cite{HolmStochasticFluids2015}, \cite{HolmUncertainty2017}), and kinetic plasma theory (\cite{KrausTyranowski2019}).

\paragraph{Main content}
The main content of the remainder of this paper is, as follows. 
\begin{description}
\item
In Section~\ref{sec:The Vlasov-Maxwell-Fokker-Planck equation} we recast the collisional Vlasov-Maxwell equations as a system of coupled stochastic and partial differential equations.
\item
In Section~\ref{sec:Particle discretization} we discuss the relationship between particle methods and stochastic modeling. We formulate a stochastic particle discretization for the collisional Vlasov-Maxwell equations and cast it in a form that allows the derivation of a variational principle.
\item
In Section~\ref{sec: Variational principle} we describe the variational structure underlying the stochastic particle discretization of the Vlasov-Maxwell system. The main result of this section is Theorem~\ref{thm: Stochastic Lagrange-d'Alembert principle for particles}, in which a stochastic Lagrange-d'Alembert principle for the particle discretization is proved. 
\item
In Section~\ref{sec: Variational principle for the Vlasov-Maxwell equation} we generalize the ideas from Section~\ref{sec: Variational principle} to the original undiscretized equations. The main result of this section is Theorem~\ref{thm: Stochastic Lagrange-d'Alembert principle}, in which a stochastic Lagrange-d'Alembert principle is proved for a class of the collisional Vlasov-Maxwell equations.
\item
In Section~\ref{sec: Variational principle for the Vlasov-Poisson equations} we prove a stochastic Lagrange-d'Alembert principle applicable to the Vlasov-Poisson equations. The main result of this section is Theorem~\ref{thm: Stochastic Lagrange-d'Alembert principle for the VP equations}.
\item
Section~\ref{sec:Summary} contains the summary of our work.
\end{description}

%%%%%%%%%%%%%%%%%%%%%%%%%%%%%%%%%%%%%%%%%%%%%%%%%%%%%%%%%%%%%%%%%%%%%%%%%%%%%%%%%%%%
%  The Vlasov-Maxwell-Fokker-Planck equation
%%%%%%%%%%%%%%%%%%%%%%%%%%%%%%%%%%%%%%%%%%%%%%%%%%%%%%%%%%%%%%%%%%%%%%%%%%%%%%%%%%%%
\section{The Vlasov-Maxwell-Fokker-Planck equations}
\label{sec:The Vlasov-Maxwell-Fokker-Planck equation}

\subsection{Stochastic reformulation}
\label{sec: Stochastic reformulation}

Various collision models and various forms of the collision operator $C[f]$ are considered in the plasma physics literature (see, e.g., \cite{CallenPlasmaKinetic}, \cite{MontgomeryPlasmaKinetic}). A key step towards a stochastic variational principle is a probabilistic interpretation of the Vlasov equation~\eqref{eq: Collisional Vlasov equation}. Therefore, in this work we will be interested only in those collision operators for which \eqref{eq: Collisional Vlasov equation} takes the form of a linear or strongly nonlinear Fokker-Planck equation (see, e.g., \cite{FrankNonlinearFokkerPlanck}, \cite{GardinerStochastic}, \cite{RiskenFokkerPlanck}). Namely, we will assume that the collision operator can be expressed as

\begin{equation}
\label{eq:Collision operator}
C[f] = \frac{1}{2}\sum_{i,j=1}^3 \frac{\partial^2}{\partial v^i \partial v^j}\big[D_{ij}(\mathbf{x},\mathbf{v}; f)f\big] - \sum_{i=1}^3 \frac{\partial}{\partial v^i}\big[K_{i}(\mathbf{x},\mathbf{v}; f)f\big],
\end{equation}

\noindent
for some symmetric positive semi-definite matrix $D_{ij}(\mathbf{x},\mathbf{v}; f)$ and vector $K_{i}(\mathbf{x},\mathbf{v}; f)$ functions, where the dependence of $D_{ij}$ and $K_i$ on $f$ may in general be nonlinear, and may involve differential and integral forms of $f$. In that case \eqref{eq: Collisional Vlasov equation} is an integro-differential equation, the so-called strongly nonlinear Fokker-Planck equation (see \cite{FrankNonlinearFokkerPlanck}). In case $D_{ij}$ and $K_i$ are independent of $f$, that is, $D_{ij}(\mathbf{x},\mathbf{v}; f)=D_{ij}(\mathbf{x},\mathbf{v})$ and $K_{i}(\mathbf{x},\mathbf{v}; f)=K_{i}(\mathbf{x},\mathbf{v})$, the Vlasov equation \eqref{eq: Collisional Vlasov equation} reduces to the standard linear Fokker-Planck equation. We will further assume that $D_{ij}$ and $K_i$ can be expressed in the form

\begin{align}
\label{eq:The D and K functions in terms of the forcing terms}
D_{ij}(\mathbf{x},\mathbf{v}; f) = \sum_{\nu=1}^M g^i_\nu g^j_\nu, \qquad \qquad K_i(\mathbf{x},\mathbf{v}; f) = G^i + \frac{1}{2}\sum_{\nu=1}^M \sum_{j=1}^3 \frac{\partial g^i_\nu}{\partial v^j} g^j_\nu,
\end{align}

\noindent
for a vector function $\mathbf{G}(\mathbf{x},\mathbf{v}; f)$, and a family of vector functions $\mathbf{g}_\nu(\mathbf{x},\mathbf{v}; f)$ with $\nu=1,\ldots,M$. Note that given a symmetric positive semi-definite matrix $D_{ij}$, a decomposition \eqref{eq:The D and K functions in terms of the forcing terms} can always be found, but it may not be unique. For instance, one may take $M=3$ and assume that $g^i_\nu = g^\nu_i$ for $i,\nu=1,2,3$. Then the first equation in \eqref{eq:The D and K functions in terms of the forcing terms} implies that the family of functions $g^i_\nu$ can be determined by calculating the square root of the matrix $D_{ij}$, and the second equation in \eqref{eq:The D and K functions in terms of the forcing terms} can be used to calculate the function $\mathbf{G}$. If \eqref{eq: Collisional Vlasov equation} has the form of a Fokker-Planck equation, then the particle density function $f$ can be interpreted as the probability density function for a stochastic process $(\mathbf{X}(t),\mathbf{V}(t))\in \mathbb{R}^3 \times \mathbb{R}^3$. This stochastic process then satisfies the Stratonovich stochastic differential equation (see \cite{FrankNonlinearFokkerPlanck}, \cite{GardinerStochastic}, \cite{KloedenPlatenSDE}, \cite{RiskenFokkerPlanck})

\begin{subequations}
\label{eq: SDE for X and V}
\begin{align}
\label{eq: SDE for X and V 1}
d\mathbf{X} &= \mathbf{V} \, dt, \\
\label{eq: SDE for X and V 2}
d\mathbf{V} &= \bigg(\frac{q}{m}\mathbf{E}(\mathbf{X},t)+\frac{q}{m}\mathbf{V}\times\mathbf{B}(\mathbf{X},t)+\mathbf{G}(\mathbf{X},\mathbf{V}; f) \bigg) \, dt + \sum_{\nu=1}^M \mathbf{g}_\nu (\mathbf{X},\mathbf{V}; f)\circ dW^\nu(t),
\end{align}
\end{subequations}

\noindent
where $W^1(t), \ldots, W^M(t)$ denote the components of the standard $M$-dimensional Wiener process, and $\circ$ denotes Stratonovich integration. Note that the terms $\mathbf{G}$ and $\mathbf{g}_\nu$ can be interpreted as external forces, and that in their absence the equations \eqref{eq: SDE for X and V} reduce to the equations of motion of a charged particle in an electromagnetic field. We will therefore refer to $\mathbf{G}$ and $\mathbf{g}_\nu$ as forcing terms. The electric and magnetic fields are coupled via the Maxwell equations \eqref{eq: Maxwell equations}. It should also be noted that unless \eqref{eq: Collisional Vlasov equation} is linear, the right-hand side of \eqref{eq: SDE for X and V} depends on $f$. In order to obtain a self-consistent system, one can express $f$ in terms of the stochastic processes $\mathbf{X}$ and $\mathbf{V}$ as $f(\mathbf{x},\mathbf{v},t)=\mathbb{E}[\delta(\mathbf{x}-\mathbf{X}(t))\delta(\mathbf{v}-\mathbf{V}(t))]$, where $\mathbb{E}$ denotes the expected value, and $\delta$ is Dirac's delta. This can be further plugged into \eqref{eq: Charge and current densities}. Together, we get 

\begin{subequations}
\label{eq: Charge and current densities in terms of expected values}
\begin{align}
\label{eq: Charge and current densities in terms of expected values 1}
f(\mathbf{x},\mathbf{v},t)&=\mathbb{E}[\delta(\mathbf{x}-\mathbf{X}(t))\delta(\mathbf{v}-\mathbf{V}(t))], \\
\label{eq: Charge and current densities in terms of expected values 2}
\rho(\mathbf{x},t) &= q N_{tot}\mathbb{E}[\delta(\mathbf{x}-\mathbf{X}(t))], \\
\label{eq: Charge and current densities in terms of expected values 3}
\mathbf{J}(\mathbf{x},t) &= q N_{tot} \mathbb{E}[\mathbf{V}(t)\delta(\mathbf{x}-\mathbf{X}(t))].
\end{align}
\end{subequations}

\noindent
Equations \eqref{eq: Maxwell equations}, \eqref{eq: SDE for X and V}, and \eqref{eq: Charge and current densities in terms of expected values} form a self-consistent system of stochastic and partial differential equations whose solutions are the stochastic processes $\mathbf{X}(t)$, $\mathbf{V}(t)$, and the functions $\mathbf{E}(\mathbf{x},t)$, $\mathbf{B}(\mathbf{x},t)$.

\paragraph{Remark.} Upon substituting \eqref{eq: Charge and current densities in terms of expected values 1}, the forcing terms $\mathbf{G}$ and $\mathbf{g}_\nu$ become functionals of the processes $\mathbf{X}$ and $\mathbf{V}$, that is, $\mathbf{G}(\mathbf{x},\mathbf{v}; f) = \mathbf{G}(\mathbf{x},\mathbf{v}; \mathbf{X}, \mathbf{V})$ and $\mathbf{g}_\nu(\mathbf{x},\mathbf{v}; f) = \mathbf{g}_\nu(\mathbf{x},\mathbf{v}; \mathbf{X}, \mathbf{V})$. However, for convenience and simplicity, throughout this work we will stick to the notation $\mathbf{G}(\mathbf{x},\mathbf{v}; f)$ and $\mathbf{g}_\nu(\mathbf{x},\mathbf{v}; f)$, understanding that the probability density is given by \eqref{eq: Charge and current densities in terms of expected values 1} (or by \eqref{eq: Charge and current densities from the law of large numbers 1} for particle discretizations; see Section~\ref{sec:Particle discretization}).

\subsection{Examples}
\label{sec: Examples}

Below we list a few examples of collision operators that fit the decription presented in Section~\ref{sec: Stochastic reformulation}.

\subsubsection{Lenard-Bernstein operator}
\label{sec: Lenard-Bernstein operator}

The Lenard-Bernstein collision operator,

\begin{equation}
\label{eq:Lenard-Bernstein collision operator}
C[f]=\nu_\text{c} \bigg( \mu \nabla_v \cdot (\mathbf{v} f) + \frac{\gamma^2}{2}\Delta_v f \bigg),
\end{equation}

\noindent
where $\nu_\text{c}>0$, $\mu>0$, and $\gamma>0$ are parameters, models small-angle collisions, and was originally used to study longitudinal plasma oscillations (see~\cite{CallenPlasmaKinetic}, \cite{LenardBernstein1958}, \cite{MontgomeryPlasmaKinetic}). It can be easily verified that an example decomposition \eqref{eq:The D and K functions in terms of the forcing terms} for $M=3$ is given by the functions

\begin{equation}
\label{eq: Forcing terms for Lenard-Bernstein}
\mathbf{G}(\mathbf{x}, \mathbf{v})=-\nu_\text{c} \mu \mathbf{v}, \quad \quad
\mathbf{g}_1(\mathbf{x}, \mathbf{v})=
\begin{pmatrix}
\sqrt{\nu_\text{c}}\gamma \\
0 \\
0
\end{pmatrix}, \quad \quad
\mathbf{g}_2(\mathbf{x}, \mathbf{v})=
\begin{pmatrix}
0 \\
\sqrt{\nu_\text{c}}\gamma \\
0
\end{pmatrix}, \quad \quad
\mathbf{g}_3(\mathbf{x}, \mathbf{v})=
\begin{pmatrix}
0 \\
0\\
\sqrt{\nu_\text{c}}\gamma
\end{pmatrix}.
\end{equation}

\noindent
Note that these functions do not explicitly depend on $f$, therefore in this case \eqref{eq: Collisional Vlasov equation} is a linear Fokker-Planck equation.

\subsubsection{Lorentz operator}
\label{sec: Lorentz operator}

The Lorentz collision operator models electron-ion interactions via pitch-angle scattering and is given by the formula

\begin{equation}
\label{eq: Lorentz collision operator}
C[f]=\frac{\nu_\text{c} (|\mathbf{v}|)}{2} \nabla_v \cdot \big(|\mathbf{v}|^2 \mathbb{I}-\mathbf{v}\otimes \mathbf{v}\big) \nabla_v f,
\end{equation}

\noindent
where $\nu_\text{c} (|\mathbf{v}|)$ is the collisional frequency as a function of the absolute value of velocity, $\mathbb{I}$ is the $3\times 3$ identity matrix, and $\otimes$ denotes tensor product. The primary effect of this type of scattering is a change of the direction of the electron's velocity with negligible energy loss. More information about the Lorentz collision operator, including the exact form of the collision frequency, can be found in, e.g., \cite{Banks2016}, \cite{CallenPlasmaKinetic}, \cite{Karney1986}, \cite{MontgomeryPlasmaKinetic}. It can be verified by a straightforward calculation that an example decomposition \eqref{eq:The D and K functions in terms of the forcing terms} for $M=3$ is given by the functions

\begin{equation}
\label{eq: Forcing terms for Lorentz}
\mathbf{G}(\mathbf{x}, \mathbf{v})=0, \quad
\mathbf{g}_1(\mathbf{x}, \mathbf{v})= \sqrt{\nu_\text{c} (|\mathbf{v}|)}
\begin{pmatrix}
0 \\
-v^3 \\
v^2
\end{pmatrix}, \quad
\mathbf{g}_2(\mathbf{x}, \mathbf{v})= \sqrt{\nu_\text{c} (|\mathbf{v}|)}
\begin{pmatrix}
v^3 \\
0 \\
-v^1
\end{pmatrix}, \quad
\mathbf{g}_3(\mathbf{x}, \mathbf{v})= \sqrt{\nu_\text{c} (|\mathbf{v}|)}
\begin{pmatrix}
-v^2 \\
v^1\\
0
\end{pmatrix}.
\end{equation}

\noindent
Note that these functions do not explicitly depend on $f$, therefore also in this case \eqref{eq: Collisional Vlasov equation} is a linear Fokker-Planck equation.

\subsubsection{Coulomb/Landau operator}
\label{sec: Coulomb operator}

The more general Coulomb collision operator has the form \eqref{eq:Collision operator} with

\begin{align}
\label{eq: D and K for Coulomb operator}
D_{ij}(\mathbf{x},\mathbf{v}; f)&=N_{tot}\Gamma \int_{\mathbb{R}^3}\frac{|\mathbf{v}-\mathbf{u}|^2 \delta_{ij}-(v^i-u^i)(v^j-u^j)}{|\mathbf{v}-\mathbf{u}|^3} f(\mathbf{x},\mathbf{u},t)\,d^3 \mathbf{u}, \nonumber \\
K_{i}(\mathbf{x},\mathbf{v}; f)&=-2N_{tot}\Gamma \int_{\mathbb{R}^3}\frac{v^i-u^i}{|\mathbf{v}-\mathbf{u}|^3} f(\mathbf{x},\mathbf{u},t)\,d^3 \mathbf{u},
\end{align}

\noindent
where $N_{tot}$ appears due to our normalization of $f$, $\delta_{ij}$ is Kronecker's delta, and $\Gamma=(4 \pi q^4/m^2)\ln \Lambda$, with $\ln \Lambda$ denoting the so-called Coulomb logarithm. The Coulomb operator describes collisions in which the fundamental two-body force obeys an inverse square law, and makes the assumption that small-angle collisions are more important that collisions resulting in large momentum changes (see \cite{CallenPlasmaKinetic}, \cite{MontgomeryPlasmaKinetic}, \cite{Rosenbluth1957}). A decomposition \eqref{eq:The D and K functions in terms of the forcing terms} can be found, for example, via the procedure outlined in Section~\ref{sec: Stochastic reformulation}. However, the expressions for $\mathbf{G}$ and $\mathbf{g}_\nu$ are complicated, therefore we are not stating them here explicitly. Note that $D_{ij}$ and $K_i$ explicitly depend on $f$. Therefore, for the Coulomb operator the Vlasov equation \eqref{eq: Collisional Vlasov equation} is a strongly nonlinear Fokker-Planck equation. Note also that $D_{ij}$ and $K_i$ can be explicitly written as functionals of the stochastic processes $\mathbf{X}$ and $\mathbf{V}$ as

\begin{align}
\label{eq: D and K for Coulomb operator - as functionals of X and V}
D_{ij}(\mathbf{x},\mathbf{v}; \mathbf{X}, \mathbf{V})&=N_{tot}\Gamma \cdot\mathbb{E}\bigg[\frac{|\mathbf{v}-\mathbf{V}(t)|^2 \delta_{ij}-(v^i-V^i(t))(v^j-V^j(t))}{|\mathbf{v}-\mathbf{V}(t)|^3} \delta(\mathbf{x}-\mathbf{X}(t))\bigg], \nonumber \\
K_{i}(\mathbf{x},\mathbf{v}; \mathbf{X}, \mathbf{V})&=-2N_{tot}\Gamma\cdot \mathbb{E}\bigg[\frac{v^i-V^i(t)}{|\mathbf{v}-\mathbf{V}(t)|^3} \delta(\mathbf{x}-\mathbf{X}(t))\bigg].
\end{align}

\noindent
The collision operator \eqref{eq:Collision operator} with $D_{ij}$ and $K_i$ as in \eqref{eq: D and K for Coulomb operator} can also be expressed in an equivalent, although more symmetric form, known as the Landau form of the Coulomb operator, or simply the Landau collision operator (see, e.g., \cite{CallenPlasmaKinetic}).

%%%%%%%%%%%%%%%%%%%%%%%%%%%%%%%%%%%%%%%%%%%%%%%%%%%%%%%%%%%%%%%%%%%%%%%%%%%%%%%%%%%%
%  Stochastic particle discretization
%%%%%%%%%%%%%%%%%%%%%%%%%%%%%%%%%%%%%%%%%%%%%%%%%%%%%%%%%%%%%%%%%%%%%%%%%%%%%%%%%%%%
\section{Stochastic particle discretization of the Vlasov-Maxwell equations}
\label{sec:Particle discretization}

Particle modelling is one of the most popular numerical techniques for solving the Vlasov equation (see, e.g., \cite{BirdsallPlasma}, \cite{HockneyParticles}). In this section we discuss the connections between particle methods and stochastic modelling.

The standard particle method for the collisionless Vlasov equation \eqref{eq: Collisional Vlasov equation} (with $C[f]=0$) consists of substituting the Ansatz $f(\mathbf{x},\mathbf{v},t) = \sum_{a=1}^N w_a \delta(\mathbf{x}-\mathbf{X}_a(t))\delta(\mathbf{v}-\mathbf{V}_a(t))$ for the particle density function, and deriving the corresponding ordinary differential equations satisfied by the \textquoteleft particle' positions $\mathbf{X}_a(t)$ and velocities $\mathbf{V}_a(t)$, which turn out to be the characteristic equations. Note that we did a qualitatively similar thing in Section~\ref{sec: Stochastic reformulation}, where we turned the original collisional Vlasov equation into the system of stochastic differential equations \eqref{eq: SDE for X and V}, which in the absence of the forcing terms $\mathbf{G}$ and $\mathbf{g}_\nu$ have the same form as the characteristic equations, and in fact the \textquoteleft particles' $\mathbf{X}_a(t)$ and $\mathbf{V}_a(t)$ can be interpreted as realizations of the stochastic processes $\mathbf{X}(t)$ and $\mathbf{V}(t)$ for different elementary events $\omega \in \Omega$.

When the right-hand side of \eqref{eq: SDE for X and V} does not depend on $f$, then \eqref{eq: SDE for X and V} can in principle be solved numerically with the help of any standard stochastic numerical method (see, e.g., \cite{KloedenPlatenSDE}), and each realization of the stochastic processes can be simulated independently of others. When the right-hand side of \eqref{eq: SDE for X and V} depends on $f$, then all realizations of the stochastic processes have to be solved for simultaneously, so that at each time step the probability density function $f$ can be numerically approximated (see, e.g., \cite{FrankNonlinearFokkerPlanck}). Such an approach, however, does not quite lend itself to a geometric formulation. Therefore, in order to be able to introduce a variational principle in Section~\ref{sec: Variational principle}, let us consider $2N$ stochastic processes $\mathbf{X}_1, \mathbf{V}_1, \ldots, \mathbf{X}_N, \mathbf{V}_N$, with each pair $(\mathbf{X}_a, \mathbf{V}_a)$ satisfying the stochastic differential system

\begin{subequations}
\label{eq: SDE for X_a and V_a}
\begin{align}
\label{eq: SDE for X_a and V_a 1}
d\mathbf{X}_a &= \mathbf{V}_a \, dt, \\
\label{eq: SDE for X_a and V_a 2}
d\mathbf{V}_a &= \bigg(\frac{q}{m}\mathbf{E}(\mathbf{X}_a,t)+\frac{q}{m}\mathbf{V}_a\times\mathbf{B}(\mathbf{X}_a,t)+\mathbf{G}(\mathbf{X}_a,\mathbf{V}_a; f) \bigg) \, dt + \sum_{\nu=1}^M \mathbf{g}_\nu (\mathbf{X}_a,\mathbf{V}_a; f)\circ dW^\nu_a(t),
\end{align}
\end{subequations}

\noindent
for $a=1,\ldots,N$, where $\mathbf{W}_a=(W_a^1,\ldots,W_a^M)$ are $N$ independent $M$-dimensional Wiener processes. Note that the systems \eqref{eq: SDE for X_a and V_a} are decoupled from each other for different values of $a$, and each system is driven by an independent Wiener process $\mathbf{W}_a$. Therefore, the pairs $(\mathbf{X}_a, \mathbf{V}_a)$ for $a=1,\ldots,N$ are independent identically distributed (i.i.d.) stochastic processes, each with the probability density function $f$ that satisfies the original Fokker-Planck equation \eqref{eq: Collisional Vlasov equation}. In that sense \eqref{eq: SDE for X_a and V_a} is equivalent to \eqref{eq: SDE for X and V}. The advantage is that instead of considering $N$ realizations of the 6-dimensional stochastic process $(\mathbf{X}, \mathbf{V})$ in \eqref{eq: SDE for X and V}, one can consider one realization of the $6N$-dimensional process $(\mathbf{X}_1, \mathbf{V}_1, \ldots, \mathbf{X}_N, \mathbf{V}_N)$ in \eqref{eq: SDE for X_a and V_a}. Such a reformulation will allow us to identify an underlying stochastic variational principle in Section~\ref{sec: Variational principle}. The last step leading to the stochastic particle discretization is approximating the probability density function $f$ in \eqref{eq: SDE for X_a and V_a}. This can be done with the help of the law of large numbers, namely, one can approximate \eqref{eq: Charge and current densities in terms of expected values} for large $N$ as

\begin{subequations}
\label{eq: Charge and current densities from the law of large numbers}
\begin{align}
\label{eq: Charge and current densities from the law of large numbers 1}
f(\mathbf{x},\mathbf{v},t)&\approx\frac{1}{N}\sum_{a=1}^N\delta(\mathbf{x}-\mathbf{X}_a(t))\delta(\mathbf{v}-\mathbf{V}_a(t)), \\
\label{eq: Charge and current densities from the law of large numbers 2}
\rho(\mathbf{x},t) &\approx \frac{q N_{tot}}{N}\sum_{a=1}^N\delta(\mathbf{x}-\mathbf{X}_a(t)), \\
\label{eq: Charge and current densities from the law of large numbers 3}
\mathbf{J}(\mathbf{x},t) &\approx \frac{q N_{tot}}{N} \sum_{a=1}^N\mathbf{V}_a(t)\delta(\mathbf{x}-\mathbf{X}_a(t)).
\end{align}
\end{subequations}

\noindent
It is easy to see that \eqref{eq: Charge and current densities from the law of large numbers 1} coincides with the standard Ansatz used in particle modelling (with the weights $w_a=1/N$). Therefore, the system of stochastic differential equations \eqref{eq: SDE for X_a and V_a} with the approximation \eqref{eq: Charge and current densities from the law of large numbers}, and with the electromagnetic field coupled via the Maxwell equations \eqref{eq: Maxwell equations}, can be considered as a stochastic particle discretization of the collisional Vlasov-Maxwell equations.

\paragraph{Remark.} Upon substituting \eqref{eq: Charge and current densities from the law of large numbers 1}, the forcing terms $\mathbf{G}$ and $\mathbf{g}_\nu$ become functionals of the processes $\mathbf{X}_1, \ldots, \mathbf{X}_N$ and $\mathbf{V}_1, \ldots, \mathbf{V}_N$. Similar to the discussion in Section~\ref{sec: Stochastic reformulation}, for convenience and simplicity, throughout this work we will stick to the notation $\mathbf{G}(\mathbf{x},\mathbf{v}; f)$ and $\mathbf{g}_\nu(\mathbf{x},\mathbf{v}; f)$, understanding that the probability density is given by \eqref{eq: Charge and current densities from the law of large numbers 1} for particle discretizations.

%%%%%%%%%%%%%%%%%%%%%%%%%%%%%%%%%%%%%%%%%%%%%%%%%%%%%%%%%%%%%%%%%%%%%%%%%%%%%%%%%%%%
%  Variational principle for particle discretizations
%%%%%%%%%%%%%%%%%%%%%%%%%%%%%%%%%%%%%%%%%%%%%%%%%%%%%%%%%%%%%%%%%%%%%%%%%%%%%%%%%%%%
\section{Variational principle for the particle discretization}
\label{sec: Variational principle}

In this section we propose an action functional which can be understood as a stochastic version of the Low action functional (see \cite{Low1958Lagrangian}), and we prove a variational principle underlying the particle discretization introduced in Section~\ref{sec:Particle discretization}, akin to the stochastic Lagrange-d'Alembert principle first introduced in \cite{KrausTyranowski2019}.

\subsection{Function spaces}
\label{eq: Function spaces}
Before we introduce the action functional, we need to identify suitable function spaces on which it will be defined. For simplicity, let our spatial domain be the whole three-dimensional space $\mathbb{R}^3$, and let us consider the time interval $[0,T]$ for some $T>0$. Let $(\Omega, \mathcal{F},\mathbb{P})$ be the probability space with the filtration $\{\mathcal{F}_t\}_{t \geq 0}$, and let $\mathbf{W}_a=(W_a^1,\ldots,W_a^M)$ for $a=1,\ldots,N$ denote $N$ independent $M$-dimensional Wiener processes on that probability space (such that $W_a^\nu(t)$ is $\mathcal{F}_t$-measurable for all $t\geq 0$). The stochastic processes $\mathbf{X}_a(t)$ and $\mathbf{V}_a(t)$ satisfy \eqref{eq: SDE for X_a and V_a}, so they are in particular $\mathcal{F}_t$-adapted semimartingales, and have almost surely continuous paths (see \cite{ProtterStochastic}). We also notice that there is no diffusion term in \eqref{eq: SDE for X_a and V_a 1}, therefore we even have that the processes $\mathbf{X}_a(t)$ are almost surely of class $C^1$. We introduce the notation

\begin{equation}
\label{eq: Space of stochastic processes}
C^k_{\Omega,T}=\big\{ \mathbf{X}\in L^2(\Omega \times [0, T], \mathbb{R}^3)  \, \big| \, \text{$\mathbf{X}$ is a $\mathcal{F}_t$-adapted semimartingale, almost surely of class $C^k$} \big\}.
\end{equation}

\noindent
Note that this set is a vector space (see \cite{ProtterStochastic}). The potentials $\varphi$ and $\mathbf{A}$ satisfy the Maxwell equations \eqref{eq: Maxwell equations} and \eqref{eq: Scalar and vector potentials}, therefore we require them to be of class $C^2$. However, since our spatial domain is unbounded, we further need to assume that the vector fields $\mathbf{E}$ and $\mathbf{B}$ are square integrable. We introduce the notation

\begin{align}
\label{eq: Space of potentials}
\mathfrak{X}(\mathbb{R}^n)&=\big\{ \mathbf{A} \in C^2(\mathbb{R}^3\times [0, T],\mathbb{R}^n) \cap L^{\infty}(\mathbb{R}^3\times [0, T],\mathbb{R}^n)  \, \big| \, \text{$\forall i,j: \frac{\partial A^i}{\partial x^j}, \frac{\partial A^i}{\partial t} \in L^2(\mathbb{R}^3\times [0, T])$} \big\}, \nonumber \\
\mathfrak{X}_0(\mathbb{R}^n)&=C^2_0(\mathbb{R}^3\times [0, T],\mathbb{R}^n),
\end{align} 

\noindent
where $\mathfrak{X}_0(\mathbb{R}^n)$ is simply the space of compactly supported elements of $\mathfrak{X}(\mathbb{R}^n)$.

\subsection{Action functional}
\label{sec: Action functional}

Let us consider the action functional

\begin{equation}
\label{eq: Action functional for particles - domain}
S: \Omega \times \big(C^1_{\Omega,T}\big)^N \times \big(C^0_{\Omega,T}\big)^N \times \big(C^0_{\Omega,T}\big)^N \times \mathfrak{X}(\mathbb{R}) \times \mathfrak{X}(\mathbb{R}^3) \longrightarrow \mathbb{R}
\end{equation}

\noindent
defined by the formula

\begin{align}
\label{eq: Action functional for particles - formula}
S[&\mathbf{X}_1,\ldots,\mathbf{X}_N,\mathbf{V}_1,\ldots,\mathbf{V}_N,\mathbf{P}_1,\ldots,\mathbf{P}_N,\varphi,\mathbf{A}] = \nonumber\\
&\frac{N_{tot}}{N}\sum_{a=1}^N\Bigg[ \int_0^T \bigg( \frac{m}{2}|\mathbf{V}_a|^2-q \varphi(\mathbf{X}_a,t)+q\mathbf{V}_a\cdot\mathbf{A}(\mathbf{X}_a,t)+\mathbf{P}_a\cdot(\mathbf{\dot X}_a-\mathbf{V}_a)\bigg)\,dt \Bigg]+\int_0^T\!\!\!\!\! \int_{\mathbb{R}^3}\frac{1}{2}(|\mathbf{E}|^2-|\mathbf{B}|^2)\,d^3\mathbf{x}dt,
\end{align}

\noindent
where $\mathbf{\dot X}_a$ denotes the time derivative of $\mathbf{X}_a$, and the electric and magnetic fields $\mathbf{E}$ and $\mathbf{B}$ are expressed in terms of the partial derivatives of the potentials $\varphi$ and $\mathbf{A}$ as in \eqref{eq: Scalar and vector potentials}. Following the standard convention in stochastic analysis, we will omit writing elementary events $\omega \in \Omega$ as arguments of stochastic processes unless otherwise needed, i.e., $\mathbf{X}_a(t) \equiv \mathbf{X}_a(\omega, t)$. The action functional \eqref{eq: Action functional for particles - formula} resembles the Low action functional introduced in \cite{Low1958Lagrangian}. In fact, it can be viewed as a particle discretization of the Low action functional, written in terms of stochastic processes (see \cite{Evstatiev2013}, \cite{Lewis1970}, \cite{Shadwick2014}, \cite{Squire2012}, \cite{XiaoLiuQin2013}, \cite{XiaoQinLiu2018}). The term $\mathbf{P}_a\cdot(\mathbf{\dot X}_a-\mathbf{V}_a)$ is the so-called Hamilton-Pontryagin kinematic constraint (see, e.g., \cite{LallWestHamiltonian}, \cite{YoshimuraMarsdenDiracStructures2}) that enforces that $\mathbf{\dot X}_a=\mathbf{V}_a$ using the Lagrange multiplier $\mathbf{P}_a$, which turns out to be the conjugate momentum. In principle, this constraint is not necessary in our context---we could omit it and replace $\mathbf{V}_a$ with $\mathbf{\dot X}_a$ in \eqref{eq: Action functional for particles - formula}. We will, however, keep it in order to make a clear connection with the theory developed in \cite{BouRabeeSVI}. It also makes the notation in the proof of the stochastic Lagrange-d'Alembert principle in Section~\ref{sec: The stochastic Lagrange-d'Alembert principle} more convenient and elegant. Note that the action functional $S$ is itself a random variable, as $\omega \in \Omega$ is one of its arguments. The variations of $S$ with respect to its arguments are given by (see Appendix~\ref{sec: The variations of the action functional S} for the details of the derivations)

\begin{subequations}
\begin{align}
\label{eq: Variation of S wrt X - calculation 2}
\delta_{\mathbf{X}_a}S =&\frac{N_{tot}}{N}\bigg(\mathbf{P}_a(T) \cdot \delta \mathbf{X}_a(T)-\mathbf{P}_a(0) \cdot \delta \mathbf{X}_a(0) \bigg) \nonumber \\
&+\frac{N_{tot}}{N}\Bigg[ - \int_0^T \delta \mathbf{X}_a\circ d\mathbf{P}_a + \int_0^T \bigg( -q \nabla_x \varphi(\mathbf{X}_a,t)\cdot \delta\mathbf{X}_a+q\sum_{i,j=1}^3V^j\frac{\partial A^j}{\partial x^i}(\mathbf{X}_a,t)\delta X^i_a\bigg)\,dt \Bigg], \\
\label{eq: Variation of S wrt V}
\delta_{\mathbf{V}_a}S =& \frac{N_{tot}}{N}\int_0^T \big( m\mathbf{V}_a+q\mathbf{A}(\mathbf{X}_a,t)-\mathbf{P}_a\big)\cdot \delta\mathbf{V}_a\,dt, \\
\label{eq: Variation of S wrt P}
\delta_{\mathbf{P}_a}S =& \frac{N_{tot}}{N}\int_0^T \big( \mathbf{\dot X}_a -\mathbf{V}_a\big)\cdot \delta\mathbf{P}_a\,dt, \\
\label{eq: Variation of S wrt A - calculation}
\delta_{\mathbf{A}}S =& \int_0^T\!\!\!\!\! \int_{\mathbb{R}^3}\bigg(\mathbf{J}+\frac{\partial \mathbf{E}}{\partial t}-\nabla_x\times \mathbf{B}\bigg)\cdot \delta \mathbf{A}\,d^3\mathbf{x}dt - \int_{\mathbb{R}^3}\Big( \mathbf{E}(\mathbf{x},T)\cdot\delta\mathbf{A}(\mathbf{x},T) - \mathbf{E}(\mathbf{x},0)\cdot\delta\mathbf{A}(\mathbf{x},0) \Big)\,d^3\mathbf{x},\\
\label{eq: Variation of S wrt phi - calculation}
\delta_{\varphi}S =& \int_0^T\!\!\!\!\! \int_{\mathbb{R}^3}\big(\nabla_x\cdot\mathbf{E}-\rho\big)\cdot \delta \varphi\,d^3\mathbf{x}dt,
\end{align}
\end{subequations}

\noindent
where $\rho$ and $\mathbf{J}$ are defined in \eqref{eq: Charge and current densities from the law of large numbers 2} and \eqref{eq: Charge and current densities from the law of large numbers 3}, respectively. The total variation of $S$ with respect to the variations of all arguments equals

\begin{equation}
\label{eq: Total variation of S}
\delta S =\sum_{a=1}^N \Big(\delta_{\mathbf{X}_a}S+\delta_{\mathbf{V}_a}S+\delta_{\mathbf{P}_a}S\Big)+\delta_{\varphi}S+\delta_{\mathbf{A}}S. 
\end{equation}

\subsection{The stochastic Lagrange-d'Alembert principle}
\label{sec: The stochastic Lagrange-d'Alembert principle}

While the standard rules of the calculus of variations apply to the variations \eqref{eq: Variation of S wrt A - calculation} and \eqref{eq: Variation of S wrt phi - calculation}, the variations \eqref{eq: Variation of S wrt X - calculation 2}, \eqref{eq: Variation of S wrt V}, \eqref{eq: Variation of S wrt P} involve stochastic processes and stochastic integrals. Therefore, before we can formulate a stochastic variational principle, we need the following lemma, whose proof is given in Appendix~\ref{sec: Proof of Lemma}.\\

\begin{lemma}
\label{thm: Fundamental lemma of calculus of variations}
Let $\mathbf{X}\in C^1_{\Omega,T}$ and $\mathbf{V}, \mathbf{P}\in C^0_{\Omega,T}$, and let $\mathbf{R},\mathbf{r}_\nu: \mathbb{R}^3 \times \mathbb{R}^3 \longrightarrow \mathbb{R}^3$ be of class $C^1$ for $\nu=1,\ldots,M$. Then

\begin{equation}
\label{eq: Lemma2 equation 1}
\forall \mathbf{Z}\in C^1_{\Omega,T}: \int_0^T \Big( \mathbf{Z}(t)\circ d\mathbf{P}-\mathbf{R}(\mathbf{X},\mathbf{V})\cdot\mathbf{Z}(t)\,dt - \sum_{\nu=1}^M\mathbf{r}_\nu(\mathbf{X},\mathbf{V})\cdot\mathbf{Z}(t)\circ dW^\nu(t) \Big) =0 \text{\quad a.s.}
\end{equation}

\noindent
if and only if

\begin{equation}
\label{eq: Lemma2 equation 2}
\forall t\in[0,T]: \int_0^t \Big(d\mathbf{P}(\tau)-\mathbf{R}(\mathbf{X}(\tau),\mathbf{V}(\tau))\,d\tau - \sum_{\nu=1}^M\mathbf{r}_\nu(\mathbf{X}(\tau),\mathbf{V}(\tau))\circ dW^\nu(\tau) \Big) =0 \text{\quad a.s.,}
\end{equation}

\noindent
where \textquoteleft a.s.' means almost surely.
\end{lemma}

\paragraph{Remark.} Equation \eqref{eq: Lemma2 equation 2} means that $\mathbf{P}(t)$, $\mathbf{X}(t)$, and $\mathbf{V}(t)$ satisfy a stochastic differential equation, which can be written in the differential form as

\begin{equation}
\label{eq: SDE for P}
d\mathbf{P}(t)=\mathbf{R}(\mathbf{X}(t),\mathbf{V}(t))\,dt + \sum_{\nu=1}^M\mathbf{r}_\nu(\mathbf{X}(t),\mathbf{V}(t))\circ dW^\nu(t).
\end{equation}

We are now in a position to formulate and prove a stochastic variational principle that generalizes the deterministic Lagrange-d'Alembert principle for forced Lagrangian and Hamiltonian systems, akin to the stochastic variational principle introduced in \cite{KrausTyranowski2019}.

\begin{thm}[{\bf Stochastic Lagrange-d'Alembert principle for particles}]
\label{thm: Stochastic Lagrange-d'Alembert principle for particles}
Let $\mathbf{X}_a\in C^1_{\Omega,T}$ and $\mathbf{V}_a, \mathbf{P}_a\in C^0_{\Omega,T}$ for $a=1,\ldots,N$ be stochastic processes, and let $\mathbf{A}\in \mathfrak{X}(\mathbb{R}^3)$, $\varphi\in \mathfrak{X}(\mathbb{R})$ be functions. Assume that $\mathbf{G}(\cdot, \cdot;f)$ and $\mathbf{g}_\nu(\cdot, \cdot;f)$ for $\nu=1,\ldots,M$ are $C^1$ functions of their arguments, where $f$ is given by \eqref{eq: Charge and current densities from the law of large numbers 1}. Then $\mathbf{X}_a$, $\mathbf{V}_a$, $\mathbf{P}_a$, $\mathbf{A}$, and $\varphi$ satisfy the system of stochastic differential equations

\begin{subequations}
\label{eq: SDEs for the Lagrange-d'Alembert principle for particles}
\begin{align}
\label{eq: SDEs for the Lagrange-d'Alembert principle for particles 1}
\mathbf{\dot X}_a &= \mathbf{V}_a, \\
\label{eq: SDEs for the Lagrange-d'Alembert principle for particles 2}
\mathbf{P}_a      &= m\mathbf{V}_a+q\mathbf{A}(\mathbf{X}_a,t), \\
\label{eq: SDEs for the Lagrange-d'Alembert principle for particles 3}
dP^i_a     &= \bigg(-q \frac{\partial \varphi}{\partial x^i}(\mathbf{X}_a,t) + q\sum_{j=1}^3V^j_a\frac{\partial A^j}{\partial x^i}(\mathbf{X}_a,t) + m \, G^i(\mathbf{X}_a, \mathbf{V}_a;f) \bigg)\,dt +  m\sum_{\nu=1}^M g^i_\nu (\mathbf{X}_a,\mathbf{V}_a; f)\circ dW^\nu_a(t),
\end{align}
\end{subequations}

\noindent
for $i=1,2,3$ and $a=1,\ldots,N$, together with the Maxwell equations \eqref{eq: Maxwell equations}, \eqref{eq: Scalar and vector potentials}, \eqref{eq: Charge and current densities from the law of large numbers} on the time interval $[0,T]$, if and only if they satisfy the following variational principle

\begin{equation}
\label{eq: stochastic Lagrange-d'Alembert principle for particles}
\delta S + \frac{mN_{tot}}{N} \sum_{a=1}^N \bigg[ \int_0^T \mathbf{G}(\mathbf{X}_a,\mathbf{V}_a; f)\cdot \delta \mathbf{X}_a \, dt + \sum_{\nu=1}^M \int_0^T\mathbf{g}_\nu (\mathbf{X}_a,\mathbf{V}_a; f)\cdot \delta \mathbf{X}_a\circ dW^\nu_a(t) \bigg] = 0
\end{equation}

\noindent
for arbitrary variations $\delta\mathbf{X}_a\in C^1_{\Omega,T}$, $\delta\mathbf{V}_a, \delta\mathbf{P}_a\in C^0_{\Omega,T}$, $\delta \mathbf{A}\in \mathfrak{X}_0(\mathbb{R}^3)$, and $\delta\varphi\in \mathfrak{X}_0(\mathbb{R})$, with  $\delta\mathbf{X}_a(0)=\delta\mathbf{X}_a(T)=0$ almost surely, and $\delta \mathbf{A}(\mathbf{x},0)=\delta \mathbf{A}(\mathbf{x},T)=0$ for all $\mathbf{x} \in \mathbb{R}^3$, where the action functional $S$ is given by \eqref{eq: Action functional for particles - formula}.
\end{thm}

\begin{proof}
Let us first consider the variations with respect to $\mathbf{A}$ in \eqref{eq: stochastic Lagrange-d'Alembert principle for particles}. Given the boundary conditions for $\delta \mathbf{A}$, from the standard calculus of variations we have that $\delta_{\mathbf{A}}S=0$ (see Equation~\eqref{eq: Variation of S wrt A - calculation}) for all $\delta \mathbf{A}$ if and only if \eqref{eq: Maxwell equations 4} is satisfied. Similarly, $\delta_{\varphi}S=0$ (see Equation~\eqref{eq: Variation of S wrt phi - calculation}) holds for all $\delta \varphi$ if and only if \eqref{eq: Maxwell equations 1} holds. Further, for variations with respect to $\mathbf{V}_a$ we have that $\delta_{\mathbf{V}_a}S=0$ (see Equation~\eqref{eq: Variation of S wrt V}) for all $\delta \mathbf{V}_a$ if and only if \eqref{eq: SDEs for the Lagrange-d'Alembert principle for particles 2} is satisfied almost surely, which follows from the standard theorem of the calculus of variations, since the integral in \eqref{eq: Variation of S wrt V} is a standard Lebesgue integral, and the integrands are almost surely continuous. Similarly, $\delta_{\mathbf{P}_a}S=0$ (see Equation~\eqref{eq: Variation of S wrt P}) for all $\delta \mathbf{P}_a$ if and only if \eqref{eq: SDEs for the Lagrange-d'Alembert principle for particles 1} is satisfied almost surely. Finally, for variations with respect to $\mathbf{X}_a$, Equations~\eqref{eq: Variation of S wrt X - calculation 2} and \eqref{eq: stochastic Lagrange-d'Alembert principle for particles} give

\begin{align}
\label{eq: Stochastic varational principle - variations wrt X}
\int_0^T \Bigg( -\delta &\mathbf{X}_a\circ d\mathbf{P}_a \nonumber \\
 &+\bigg( -q \nabla_x \varphi(\mathbf{X}_a,t)\cdot \delta\mathbf{X}_a+q\sum_{i,j=1}^3V^j\frac{\partial A^j}{\partial x^i}(\mathbf{X}_a,t)\delta X_a^i+m\mathbf{G}(\mathbf{X}_a,\mathbf{V}_a; f)\cdot \delta \mathbf{X}_a\bigg)\,dt \nonumber \\
 &\qquad\qquad\qquad\qquad\qquad\qquad\qquad\qquad\qquad\qquad+m\sum_{\nu=1}^M \mathbf{g}_\nu (\mathbf{X}_a,\mathbf{V}_a; f)\cdot \delta \mathbf{X}_a\circ dW_a^\nu(t) \Bigg) = 0,
\end{align}

\noindent
which, by Lemma~\ref{thm: Fundamental lemma of calculus of variations}, holds for all $\delta \mathbf{X}_a$ if and only if \eqref{eq: SDEs for the Lagrange-d'Alembert principle for particles 3} is satisfied.

\end{proof}

\paragraph{Remark.} Equation~\eqref{eq: SDEs for the Lagrange-d'Alembert principle for particles} is expressed in terms of the Lagrange multipliers $\mathbf{P}_a$, which, as can be seen in \eqref{eq: SDEs for the Lagrange-d'Alembert principle for particles 2}, turn out to be the conjugate momenta. The conjugate momenta can be eliminated, and Equation~\eqref{eq: SDEs for the Lagrange-d'Alembert principle for particles} can be recast as Equation~\eqref{eq: SDE for X_a and V_a 2}, which is shown in the following theorem.

\begin{thm}
\label{thm: Equivalence of the SDEs for X_a, V_a, P_a}
Equations~\eqref{eq: SDE for X_a and V_a} and \eqref{eq: SDEs for the Lagrange-d'Alembert principle for particles} are equivalent.
\end{thm}

\begin{proof}
By calculating the stochastic differential on both sides of \eqref{eq: SDEs for the Lagrange-d'Alembert principle for particles 2} and substituting \eqref{eq: SDEs for the Lagrange-d'Alembert principle for particles 1}, we obtain

\begin{equation}
\label{eq: Stochastic differential of P_a}
dP_a^i = m \,dV_a^i + q \sum_{j=1}^3 V_a^j \frac{\partial A^i}{\partial x^j}(\mathbf{X}_a,t)\,dt + q \frac{\partial A^i}{\partial t}(\mathbf{X}_a,t) \, dt
\end{equation}

\noindent
for each $i=1,2,3$ and $a=1,\ldots,N$. Comparing this with \eqref{eq: SDEs for the Lagrange-d'Alembert principle for particles 3}, and using \eqref{eq: Scalar and vector potentials}, one eliminates the conjugate momenta and obtains Equation~\eqref{eq: SDE for X_a and V_a 2}.

\end{proof}

\paragraph{Remark.} Theorems \ref{thm: Stochastic Lagrange-d'Alembert principle for particles} and \ref{thm: Equivalence of the SDEs for X_a, V_a, P_a} provide a variational formulation of the stochastic particle method from Section~\ref{sec:Particle discretization}. One can further perform a variational discretization of the electromagnetic fields $\mathbf{A}$ and $\varphi$, for instance along the lines of \cite{Squire2012}, \cite{SternDesbrun} or \cite{KrausGEMPIC}, thus obtaining a stochastic particle-in-cell (PIC) discretization of the collisional Vlasov-Maxwell equations. The resulting structure-preserving numerical methods will be investigated in a follow-up work.

%%%%%%%%%%%%%%%%%%%%%%%%%%%%%%%%%%%%%%%%%%%%%%%%%%%%%%%%%%%%%%%%%%%%%%%%%%%%%%%%%%%%
%  Variational principle for the Vlasov-Maxwell equation
%%%%%%%%%%%%%%%%%%%%%%%%%%%%%%%%%%%%%%%%%%%%%%%%%%%%%%%%%%%%%%%%%%%%%%%%%%%%%%%%%%%%
\section{Variational principle for the Vlasov-Maxwell equations}
\label{sec: Variational principle for the Vlasov-Maxwell equation}

The form of the action functional \eqref{eq: Action functional for particles - formula} and of the Lagrange-d'Alembert principle \eqref{eq: stochastic Lagrange-d'Alembert principle for particles} suggests that it should be possible to formulate a similar variational principle for the stochastic reformulation of the Vlasov-Maxwell system discussed in Section~\ref{sec: Stochastic reformulation}. In this section we provide such a variational principle for a class of collision operators.

\subsection{Action functional}
\label{sec: Action functional for the Vlasov-Maxwell system}

Let us consider the action functional defined by the formula

\begin{equation}
\label{eq: Action functional - formula}
\bar S[\mathbf{X},\mathbf{V},\mathbf{P},\varphi,\mathbf{A}] = N_{tot}\cdot\mathbb{E}\Bigg[ \int_0^T \bigg( \frac{m}{2}|\mathbf{V}|^2-q \varphi(\mathbf{X},t)+q\mathbf{V}\cdot\mathbf{A}(\mathbf{X},t)+\mathbf{P}\cdot(\mathbf{\dot X}-\mathbf{V})\bigg)\,dt \Bigg]+\int_0^T\!\!\!\!\! \int_{\mathbb{R}^3}\frac{1}{2}(|\mathbf{E}|^2-|\mathbf{B}|^2)\,d^3\mathbf{x}dt,
\end{equation}

\noindent
where $\mathbf{\dot X}$ denotes the time derivative of $\mathbf{X}$, the electric and magnetic fields $\mathbf{E}$ and $\mathbf{B}$ are expressed in terms of the partial derivatives of the potentials $\varphi$ and $\mathbf{A}$ as in \eqref{eq: Scalar and vector potentials}, and $\mathbb{E}[Y] \equiv \int_\Omega Y \,d\mathbb{P}$ denotes the expected value of the random variable $Y$. Note that unlike $S$ in \eqref{eq: Action functional for particles - formula}, the action functional $\bar S$ is not a random variable, as the dependence on $\omega \in \Omega$ is integrated out with respect to the probability measure by calculating the expected value. In fact, $S$ could be regarded as a Monte Carlo approximation of $\bar S$ when the processes $\mathbf{X}_1, \ldots, \mathbf{X}_N$ are independent and identically distributed as $\mathbf{X}$, and similarly for $\mathbf{V}$ and $\mathbf{P}$. An important issue to consider is the domain of this action functional. In a fashion similar to \eqref{eq: Action functional for particles - domain}, one may want to take as the domain the set

\begin{equation}
\label{eq: Action functional - full domain}
C^1_{\Omega,T} \times C^0_{\Omega,T} \times C^0_{\Omega,T} \times \mathfrak{X}(\mathbb{R}) \times \mathfrak{X}(\mathbb{R}^3),
\end{equation}

\noindent
on which the formula \eqref{eq: Action functional - formula} is well-defined. This domain, however, turns out to be too big, in the sense that, as will be discussed below, due to the presence of the expected value the variations of $\bar S$ do not uniquely determine the set of stochastic evolution equations \eqref{eq: SDE for X and V}. It is therefore necessary to restrict \eqref{eq: Action functional - full domain} to a smaller subspace or submanifold which is compatible with the considered collision operator. Below we will demonstrate how this can be done for a class of collision operators \eqref{eq:Collision operator} for which $D_{ij}(\mathbf{x},\mathbf{v}; f) = \text{const}$, that is, we have

\begin{equation}
\label{eq: Subclass of collision operators}
\mathbf{g}_\nu(\mathbf{x},\mathbf{v}; f) = \pmb{\chi}_\nu= \text{const}.
\end{equation}

\noindent
This class encompasses, for instance, the Lenard-Bernstein operator \eqref{eq: Forcing terms for Lenard-Bernstein}, or the more general nonlinear energy and momentum preserving Dougherty collision operator and its modifications (see \cite{Clemmow1969Electrodynamics}, \cite{Dougherty1964}, \cite{DoughertyWatson1967}, \cite{FilbetSonnendrucker2003}, \cite{Hakim2020}, \cite{KrausPHD}, \cite{OngYu1969}, \cite{Oppenheim1965}). For a given collision operator of the form \eqref{eq: Subclass of collision operators}, we define a compatible subset of $C^0_{\Omega,T}$, namely,

\begin{equation}
\label{eq: Compatible subspace}
C_{\text{col}} = \bigg\{ \mathbf{P}\in C^0_{\Omega,T}  \, \Big| \, \exists \mathbf{Z} \in C^0_{\Omega,T}: d\mathbf{P}=\mathbf{Z}\,dt+m\sum_{\nu=1}^M\pmb{\chi}_\nu \, dW^\nu(t)\bigg\}.
\end{equation}

\noindent
Note that for any $\mathbf{P}_1, \mathbf{P}_2 \in C_{\text{col}}$ we have that $d(\mathbf{P}_1-\mathbf{P}_2) = (\mathbf{Z}_1-\mathbf{Z}_2)\,dt$, that is, $\mathbf{P}_1-\mathbf{P}_2 \in C^1_{\Omega,T}$. Therefore, the pair $(C_{\text{col}},C^1_{\Omega,T})$ is an affine subspace of $C^0_{\Omega,T}$. The action functional $\bar S$ can now be defined as

\begin{equation}
\label{eq: Action functional - restricted domain}
\bar S:C^1_{\Omega,T} \times C^0_{\Omega,T} \times C_{\text{col}} \times \mathfrak{X}(\mathbb{R}) \times \mathfrak{X}(\mathbb{R}^3) \longrightarrow \mathbb{R}.
\end{equation}

\noindent
Similar to the calculations in Section~\ref{sec: Action functional}, the variations of $\bar S$ with respect to $\mathbf{V}$ and $\mathbf{P}$ are given by, respectively,

\begin{align}
\label{eq: Variation of bar S wrt V}
\delta_{\mathbf{V}}\bar S &= N_{tot}\cdot\mathbb{E}\bigg[ \int_0^T \big( m\mathbf{V}+q\mathbf{A}(\mathbf{X},t)-\mathbf{P}\big)\cdot \delta\mathbf{V}\,dt \bigg], \\
\label{eq: Variation of bar S wrt P}
\delta_{\mathbf{P}}\bar S &= N_{tot}\cdot\mathbb{E}\bigg[ \int_0^T \big( \mathbf{\dot X} -\mathbf{V}\big)\cdot \delta\mathbf{P}\,dt \bigg],
\end{align}

\noindent
except that here $\delta\mathbf{P} \in C^1_{\Omega,T}$, so that $\mathbf{P}+\epsilon\delta\mathbf{P} \in C_{\text{col}}$. For the variation of $\bar S$ with respect to $\mathbf{X}$ we have

\begin{align}
\label{eq: Variation of bar S wrt X - calculation 1}
\delta_{\mathbf{X}}\bar S =&N_{tot}\cdot\mathbb{E}\bigg(\mathbf{P}(T) \cdot \delta \mathbf{X}(T)-\mathbf{P}(0) \cdot \delta \mathbf{X}(0) \bigg) \nonumber \\
&+N_{tot}\cdot\mathbb{E}\Bigg[ - \int_0^T \delta \mathbf{X}\circ d\mathbf{P} + \int_0^T \bigg( -q \nabla_x \varphi(\mathbf{X},t)\cdot \delta\mathbf{X}+q\sum_{i,j=1}^3V^j\frac{\partial A^j}{\partial x^i}(\mathbf{X},t)\delta X^i\bigg)\,dt \Bigg].
\end{align}

\noindent
Since $\mathbf{P}\in C_{\text{col}}$, we have that $d\mathbf{P}=\mathbf{Z}\,dt+m\sum_{\nu=1}^M\pmb{\chi}_\nu \, dW^\nu(t)$. Furthermore, the variations $\delta \mathbf{X}$ are almost surely of class $C^1$, and therefore have sample paths of almost surely finite variation. Consequently, the quadratic covariation $[\pmb{\chi}_\nu\cdot\delta \mathbf{X}, W^\nu]_0^T=0$ almost surely (see \cite{ProtterStochastic}). Since the expected value of the It\^{o} integral with respect to the Wiener process is zero, we altogether have that 

\begin{equation}
\label{eq: Expected value of the Stratonovich term}
\mathbb{E}\bigg[\int_0^T \pmb{\chi}_\nu\cdot\delta \mathbf{X} \circ dW^\nu(t)\bigg]=0, \qquad\qquad \text{for all $\nu=1,\ldots,M$.} 
\end{equation}

\noindent
By plugging this in \eqref{eq: Variation of bar S wrt X - calculation 1}, we finally obtain

\begin{align}
\label{eq: Variation of bar S wrt X}
\delta_{\mathbf{X}}\bar S =&N_{tot}\cdot\mathbb{E}\bigg(\mathbf{P}(T) \cdot \delta \mathbf{X}(T)-\mathbf{P}(0) \cdot \delta \mathbf{X}(0) \bigg) \nonumber \\
&+N_{tot}\cdot\mathbb{E}\Bigg[ \int_0^T  \bigg( -\mathbf{Z}\cdot\delta \mathbf{X} -q \nabla_x \varphi(\mathbf{X},t)\cdot \delta\mathbf{X}+q\sum_{i,j=1}^3V^j\frac{\partial A^j}{\partial x^i}(\mathbf{X},t)\delta X^i\bigg)\,dt \Bigg].
\end{align}

\noindent
The variations with respect to $\mathbf{A}$ and $\varphi$ are the same as in \eqref{eq: Variation of S wrt A - calculation} and \eqref{eq: Variation of S wrt phi - calculation}, respectively, only with the charge and electric current densities given by \eqref{eq: Charge and current densities in terms of expected values} rather than \eqref{eq: Charge and current densities from the law of large numbers}. The total variation of $\bar S$ with respect to the variations of all arguments is given by

\begin{equation}
\label{eq: Total variation of bar S}
\delta \bar S =\delta_{\mathbf{X}}\bar S+\delta_{\mathbf{V}}\bar S+\delta_{\mathbf{P}}\bar S+\delta_{\varphi}\bar S+\delta_{\mathbf{A}}\bar S. 
\end{equation}

\subsection{The stochastic Lagrange-d'Alembert principle}
\label{sec: The stochastic Lagrange-d'Alembert principle for Vlasov-Maxwell}

In the following theorems we establish a variational principle for the system of equations \eqref{eq: Maxwell equations}, \eqref{eq: SDE for X and V}, and \eqref{eq: Charge and current densities in terms of expected values} for a class of collision operators with $\mathbf{g}_\nu(\mathbf{x},\mathbf{v}; f) = \pmb{\chi}_\nu  = \text{const}$ for all $\nu=1,\ldots,M$.

\begin{thm}[{\bf Stochastic Lagrange-d'Alembert principle for the VM equations}]
\label{thm: Stochastic Lagrange-d'Alembert principle}
Let $\mathbf{X}\in C^1_{\Omega,T}$, $\mathbf{V}\in C^0_{\Omega,T}$, $\mathbf{P}\in C_{\text{col}}$ be stochastic processes, and let $\mathbf{A}\in \mathfrak{X}(\mathbb{R}^3)$, $\varphi\in \mathfrak{X}(\mathbb{R})$ be functions. Assume that $\mathbf{G}(\cdot, \cdot;f)$ is a $C^1$ function of its arguments, where $f$ is given by \eqref{eq: Charge and current densities in terms of expected values 1}. Then $\mathbf{X}$, $\mathbf{V}$, $\mathbf{P}$, $\mathbf{A}$, and $\varphi$ satisfy the system of stochastic differential equations

\begin{subequations}
\label{eq: SDEs for the Lagrange-d'Alembert principle}
\begin{align}
\label{eq: SDEs for the Lagrange-d'Alembert principle 1}
\mathbf{\dot X} &= \mathbf{V}, \\
\label{eq: SDEs for the Lagrange-d'Alembert principle 2}
\mathbf{P}      &= m\mathbf{V}+q\mathbf{A}(\mathbf{X},t), \\
\label{eq: SDEs for the Lagrange-d'Alembert principle 3}
 dP^i     &= \bigg(-q \frac{\partial \varphi}{\partial x^i}(\mathbf{X},t) + q\sum_{j=1}^3V^j\frac{\partial A^j}{\partial x^i}(\mathbf{X},t) + m \, G^i(\mathbf{X}, \mathbf{V};f)\bigg)\,dt +m\sum_{\nu=1}^M\chi^i_\nu \, dW^\nu(t),
\end{align}
\end{subequations}

\noindent
for $i=1,2,3$, together with the Maxwell equations \eqref{eq: Maxwell equations}, \eqref{eq: Scalar and vector potentials}, \eqref{eq: Charge and current densities in terms of expected values} on the time interval $[0,T]$, if and only if they satisfy the following variational principle

\begin{equation}
\label{eq: stochastic Lagrange-d'Alembert principle}
\delta \bar S + mN_{tot}\cdot\mathbb{E}\bigg[ \int_0^T \mathbf{G}(\mathbf{X},\mathbf{V}; f)\cdot \delta \mathbf{X} \, dt \bigg] = 0
\end{equation}

\noindent
for arbitrary variations $\delta\mathbf{X}, \delta\mathbf{P}\in C^1_{\Omega,T}$, $\delta\mathbf{V}\in C^0_{\Omega,T}$, $\delta \mathbf{A}\in \mathfrak{X}_0(\mathbb{R}^3)$, and $\delta\varphi\in \mathfrak{X}_0(\mathbb{R})$, with  $\delta\mathbf{X}(0)=\delta\mathbf{X}(T)=0$ almost surely, and $\delta \mathbf{A}(\mathbf{x},0)=\delta \mathbf{A}(\mathbf{x},T)=0$ for all $\mathbf{x} \in \mathbb{R}^3$, where the action functional $\bar S$ is given by~\eqref{eq: Action functional - formula} and \eqref{eq: Action functional - restricted domain}.
\end{thm}

\begin{proof}
Similar to the proof of Theorem~\ref{thm: Stochastic Lagrange-d'Alembert principle for particles}, the equations $\delta_{\varphi}\bar S =0$ and $\delta_{\mathbf{A}}\bar S =0$ are equivalent to \eqref{eq: Maxwell equations 1} and \eqref{eq: Maxwell equations 4}, respectively. Note that  $C^0_{\Omega,T}$ is a subspace of $L^2(\Omega \times [0, T], \mathbb{R}^3)$, and $\langle \mathbf{Y}_1, \mathbf{Y}_2 \rangle = \mathbb{E}[\int_0^T \mathbf{Y}_1\cdot \mathbf{Y}_2\,dt]$ is an inner product on that space. Therefore, by substituting Equations \eqref{eq: Variation of bar S wrt V}, \eqref{eq: Variation of bar S wrt P}, and \eqref{eq: Variation of bar S wrt X} in Equation~\eqref{eq: stochastic Lagrange-d'Alembert principle}, and using the fact that the variations are arbitrary, we establish equivalence with Equations~\eqref{eq: SDEs for the Lagrange-d'Alembert principle 1}-\eqref{eq: SDEs for the Lagrange-d'Alembert principle 2}, as well as with the equation

\begin{equation}
\label{eq: Equation for Z in the variational principle}
Z^i = -q \frac{\partial \varphi}{\partial x^i}(\mathbf{X},t) + q\sum_{j=1}^3V^j\frac{\partial A^j}{\partial x^i}(\mathbf{X},t) + m \, G^i(\mathbf{X}, \mathbf{V};f),
\end{equation}

\noindent
for $i=1,2,3$, which in turn is equivalent to Equation~\eqref{eq: SDEs for the Lagrange-d'Alembert principle 3}, given the assumption $\mathbf{P}\in C_{\text{col}}$.\\
\end{proof}

\begin{thm}
\label{thm: Equivalence of the SDEs for X, V, P}
Equation~\eqref{eq: SDE for X and V} with $\mathbf{g}_\nu(\mathbf{x},\mathbf{v}; f) = \pmb{\chi}_\nu  = \text{const}$ for $\nu=1,\ldots,M$ and Equation~\eqref{eq: SDEs for the Lagrange-d'Alembert principle} are equivalent.
\end{thm}

\begin{proof}
Similar to the proof of Theorem~\ref{thm: Equivalence of the SDEs for X_a, V_a, P_a}, by calculating the stochastic differential on both sides of Equation~\eqref{eq: SDEs for the Lagrange-d'Alembert principle 2} and comparing with Equation~\eqref{eq: SDEs for the Lagrange-d'Alembert principle 3}, one eliminates $\mathbf{P}$ and obtains Equation~\eqref{eq: SDE for X and V 2}.\\
\end{proof}

\paragraph{Remark.} Note that the forcing terms $\mathbf{g}_\nu$ do not explicitly appear in the variational equation \eqref{eq: stochastic Lagrange-d'Alembert principle}. By comparing Theorem~\ref{thm: Stochastic Lagrange-d'Alembert principle for particles} and Theorem~\ref{thm: Stochastic Lagrange-d'Alembert principle}, one could intuitively expect that the relevant variational principle should read

\begin{equation}
\label{eq: Extension of the stochastic Lagrange-d'Alembert principle}
\delta \bar S + mN_{tot}\cdot\mathbb{E}\bigg[ \int_0^T \mathbf{G}(\mathbf{X},\mathbf{V}; f)\cdot \delta \mathbf{X} \, dt + \sum_{\nu=1}^M \int_0^T\mathbf{g}_\nu (\mathbf{X},\mathbf{V}; f)\cdot \delta \mathbf{X}\circ dW^\nu(t) \bigg] = 0.
\end{equation}

\noindent
However, due to the presence of the expected value in this equation, part or all of the information about the Stratonovich integral term is lost, as we saw in \eqref{eq: Expected value of the Stratonovich term} for instance. Therefore, if the domain \eqref{eq: Action functional - full domain} is chosen for $\bar S$, then the variational equations \eqref{eq: stochastic Lagrange-d'Alembert principle} or \eqref{eq: Extension of the stochastic Lagrange-d'Alembert principle} do not determine a unique set of stochastic differential equations that need to be satisfied by the considered stochastic processes. Consequently, it is necessary to encode the missing information about the forcing terms $\mathbf{g}_\nu$ in the definition of the action functional $\bar S$ by restricting its domain to a subset compatible with the considered collision operator. For the class of collision operators \eqref{eq: Subclass of collision operators} a suitable choice of the domain is proposed in \eqref{eq: Action functional - restricted domain}. For other collision operators appropriate domains will be nonlinear subspaces of \eqref{eq: Action functional - full domain}, and they will be investigated in a follow-up work.

%%%%%%%%%%%%%%%%%%%%%%%%%%%%%%%%%%%%%%%%%%%%%%%%%%%%%%%%%%%%%%%%%%%%%%%%%%%%%%%%%%%%
%  Variational principle for the Vlasov-Poisson equations
%%%%%%%%%%%%%%%%%%%%%%%%%%%%%%%%%%%%%%%%%%%%%%%%%%%%%%%%%%%%%%%%%%%%%%%%%%%%%%%%%%%%
\section{Variational principle for the Vlasov-Poisson equations}
\label{sec: Variational principle for the Vlasov-Poisson equations}

In the full Vlasov-Maxwell system the scalar $\varphi$ and vector $\mathbf{A}$ potentials are independent dynamic variables, and as such have to appear explicitly in the action functional alongside the stochastic processes $\mathbf{X}$, $\mathbf{V}$, and $\mathbf{P}$. In order to ensure the correct coupling between the stochastic processes and the electromagnetic field, an expected value was necessary in the definition of the action functional~\eqref{eq: Action functional - formula}. This created a difficulty in deriving a variational principle, as pointed out in the remark following Theorem~\ref{thm: Equivalence of the SDEs for X, V, P}. This difficulty can be circumvented for the Vlasov-Poisson equations because in this case the electrostatic potential $\varphi$ is uniquely determined by the stochastic process $\mathbf{X}$, as will be demonstrated below.

\subsection{The collisional Vlasov-Poisson equations}
\label{sec: The collisional Vlasov-Poisson equations}

The collisional Vlasov-Poisson equations

\begin{equation}
\label{eq: Collisional Vlasov-Poisson equation}
\frac{\partial f}{\partial t} + \mathbf{v}\cdot\nabla_x f + \frac{q}{m}\mathbf{E}\cdot \nabla_v f = C[f],
\end{equation}

\noindent
where

\begin{subequations}
\label{eq: Poisson's equation}
\begin{align}
\label{eq: Poisson's equation 1}
\mathbf{E}&=-\nabla_x \varphi, \\
\label{eq: Poisson's equation 2}
\Delta_x \varphi &= -\rho,
\end{align}
\end{subequations}

\noindent
and the charge density $\rho$ is given by \eqref{eq: Charge and current densities}, are an approximation of the Vlasov-Maxwell equations in the nonrelativistic zero-magnetic field limit. The associated stochastic differential equations take the form

\begin{subequations}
\label{eq: SDE for X and V for Vlasov-Poisson}
\begin{align}
\label{eq: SDE for X and V for Vlasov-Poisson 1}
d\mathbf{X} &= \mathbf{V} \, dt, \\
\label{eq: SDE for X and V for Vlasov-Poisson 2}
d\mathbf{V} &= \bigg(\frac{q}{m}\mathbf{E}(\mathbf{X},t)+\mathbf{G}(\mathbf{X},\mathbf{V}; f) \bigg) \, dt + \sum_{\nu=1}^M \mathbf{g}_\nu (\mathbf{X},\mathbf{V}; f)\circ dW^\nu(t).
\end{align}
\end{subequations}

\noindent
The equations \eqref{eq: Charge and current densities in terms of expected values 2}, \eqref{eq: Poisson's equation}, and \eqref{eq: SDE for X and V for Vlasov-Poisson} form a stochastic reformulation of the Vlasov-Poisson equations. A stochastic particle discretization and the corresponding stochastic variational principle can be derived just like in Sections~\ref{sec:Particle discretization} and \ref{sec: Variational principle}, respectively. Also, a variational principle analogous to the Lagrange-d'Alembert principle presented in Section~\ref{sec: Variational principle for the Vlasov-Maxwell equation} can be derived in a similar fashion. However, by doing so, one encounters the same difficulty with including the Stratonovich integral. In the case of the Vlasov-Poisson equations a different variational principle can be obtained by observing that the electrostatic potential $\varphi$ can be expressed as a functional of the stochastic process $\mathbf{X}$,

\begin{equation}
\label{eq: Electrostatic potential as a functional}
\varphi: \mathbb{R}^3 \times \mathbb{R} \times C^1_{\Omega,T} \longrightarrow \mathbb{R},
\end{equation}

\noindent
by solving Poisson's equation \eqref{eq: Poisson's equation 2}. Given the charge density function \eqref{eq: Charge and current densities in terms of expected values 2} and specific boundary conditions, the solution of Poisson's equation can be written using an appropriate Green's function for the Laplacian. Assuming the spatial domain is unbounded, the standard Green's function yields

\begin{equation}
\label{eq: Electrostatic potential - solution}
\varphi(\mathbf{x},t,\mathbf{X}) = \frac{1}{4 \pi} \int_{\mathbb{R}^3} \frac{\rho(\mathbf{y},t)}{|\mathbf{x}-\mathbf{y}|} d^3\mathbf{y} = \frac{q N_{tot}}{4 \pi}\mathbb{E}\bigg[\frac{1}{|\mathbf{x}-\mathbf{X}(t)|}\bigg].
\end{equation}

\noindent
From \eqref{eq: Poisson's equation 1} we have the electric field

\begin{equation}
\label{eq: Electric field - solution}
\mathbf{E}(\mathbf{x},t,\mathbf{X}) = \frac{q N_{tot}}{4 \pi}\mathbb{E}\bigg[\frac{\mathbf{x}-\mathbf{X}(t)}{|\mathbf{x}-\mathbf{X}(t)|^3}\bigg].
\end{equation}

\subsection{Action functional}
\label{eq: Action functional - Vlasov-Poisson}

Let us consider the action functional

\begin{equation}
\label{eq: Action functional for Vlasov-Poisson - domain}
\hat S: \Omega \times C^1_{\Omega,T} \times C^1_{\Omega,T} \times C^0_{\Omega,T}\times C^0_{\Omega,T} \longrightarrow \mathbb{R}
\end{equation}

\noindent
defined by the formula

\begin{equation}
\label{eq: Action functional for Vlasov-Poisson - formula}
\hat S[\mathbf{X},\mathbf{Y},\mathbf{V},\mathbf{P}] = \int_0^T \bigg( \frac{m}{2}|\mathbf{V}(t)|^2-q \varphi\big(\mathbf{X}(t),t,\mathbf{Y}\big)+\mathbf{P}(t)\cdot\big(\mathbf{\dot X}(t)-\mathbf{V}(t)\big)\bigg)\,dt,
\end{equation}

\noindent
where the electrostatic potential $\varphi$ is given by \eqref{eq: Electrostatic potential - solution}. Note that similar to $S$ in \eqref{eq: Action functional for particles - formula}, the functional $\hat S$ is itself random, and can be viewed as the action functional of particles represented by the process $\mathbf{X}$ which are moving in the electric field generated by particles represented by the process $\mathbf{Y}$. Similar to the calculations in Section~\ref{sec: Action functional}, the variations of $\hat S$ with respect to $\mathbf{X}$, $\mathbf{V}$, and $\mathbf{P}$ are given by, respectively,

\begin{subequations}
\begin{align}
\label{eq: Variation of hat S wrt X}
\delta_{\mathbf{X}}\hat S[\mathbf{X},\mathbf{Y},\mathbf{V},\mathbf{P}] &= \mathbf{P}(T) \cdot \delta \mathbf{X}(T)-\mathbf{P}(0) \cdot \delta \mathbf{X}(0) \nonumber \\
&\qquad\qquad -\int_0^T \delta \mathbf{X}(t) \circ d\mathbf{P}(t) +\int_0^T q \mathbf{E}\big(\mathbf{X}(t),t,\mathbf{Y}\big)\cdot \delta\mathbf{X}(t)\,dt, \\
\label{eq: Variation of hat S wrt V}
\delta_{\mathbf{V}}\hat S[\mathbf{X},\mathbf{Y},\mathbf{V},\mathbf{P}] &= \int_0^T \big( m\mathbf{V}(t)-\mathbf{P}(t)\big)\cdot \delta\mathbf{V}(t)\,dt, \\
\label{eq: Variation of hat S wrt P}
\delta_{\mathbf{P}}\hat S[\mathbf{X},\mathbf{Y},\mathbf{V},\mathbf{P}] &= \int_0^T \big( \mathbf{\dot X}(t) -\mathbf{V}(t)\big)\cdot \delta\mathbf{P}(t)\,dt,
\end{align}
\end{subequations}

\noindent
where the electric field $\mathbf{E}$ is given by \eqref{eq: Electric field - solution}. Note that we are not considering variations with respect to $\mathbf{Y}$. Let us for convenience define the joint variation of $\hat S$ with respect to $\mathbf{X}$, $\mathbf{V}$, and $\mathbf{P}$ as

\begin{equation}
\label{eq: Joint variation of hat S}
\delta_{(\mathbf{X},\mathbf{V},\mathbf{P})} \hat S = \delta_{\mathbf{X}}\hat S+\delta_{\mathbf{V}}\hat S+\delta_{\mathbf{P}}\hat S. 
\end{equation}

\subsection{The stochastic Lagrange-d'Alembert principle}
\label{sec: The stochastic Lagrange-d'Alembert principle for Vlasov-Poisson}

In the following theorem we formulate a variational principle for the system of equations \eqref{eq: Charge and current densities in terms of expected values 2}, \eqref{eq: Poisson's equation}, and \eqref{eq: SDE for X and V for Vlasov-Poisson}. Note that $\mathbf{E}(\mathbf{X}(t),t,\mathbf{X})$ is the electric field generated by a distribution of charged particles represented by the process $\mathbf{X}$ at time $t$, and evaluated at the random point $\mathbf{x}=\mathbf{X}(t)$ in space. Furthermore, the notation  $\delta_{\mathbf{X}}\hat S[\mathbf{X},\mathbf{X},\mathbf{V},\mathbf{P}]$ means that the variation of $\hat S$ is evaluated for the arguments $\mathbf{X},\mathbf{Y},\mathbf{V},\mathbf{P}$ with $\mathbf{Y}=\mathbf{X}$.

\begin{thm}[{\bf Stochastic Lagrange-d'Alembert principle for the VP equations}]
\label{thm: Stochastic Lagrange-d'Alembert principle for the VP equations}
Let $\mathbf{X}\in C^1_{\Omega,T}$ and $\mathbf{V}, \mathbf{P}\in C^0_{\Omega,T}$ be stochastic processes, and let $\varphi(\cdot,\cdot,\mathbf{X})\in \mathfrak{X}(\mathbb{R})$ be given by \eqref{eq: Electrostatic potential - solution}. Assume that $\mathbf{G}(\cdot, \cdot;f)$ and $\mathbf{g}_\nu(\cdot, \cdot;f)$ for $\nu=1,\ldots,M$ are $C^1$ functions of their arguments, where $f$ is given by \eqref{eq: Charge and current densities in terms of expected values 1}. Then $\mathbf{X}$, $\mathbf{V}$, and $\mathbf{P}$ satisfy the system of stochastic differential equations

\begin{subequations}
\label{eq: SDEs for the Lagrange-d'Alembert principle for Vlasov-Poisson}
\begin{align}
\label{eq: SDEs for the Lagrange-d'Alembert principle for Vlasov-Poisson 1}
\mathbf{\dot X}(t) &= \mathbf{V}(t), \\
\label{eq: SDEs for the Lagrange-d'Alembert principle for Vlasov-Poisson 2}
\mathbf{P}(t)      &= m\mathbf{V}(t), \\
\label{eq: SDEs for the Lagrange-d'Alembert principle for Vlasov-Poisson 3}
d\mathbf{P}(t)     &= \Big(q \mathbf{E}\big(\mathbf{X}(t),t,\mathbf{X}\big) + m \, \mathbf{G}\big(\mathbf{X}(t), \mathbf{V}(t);f\big) \Big)\,dt +  m\sum_{\nu=1}^M \mathbf{g}_\nu \big(\mathbf{X}(t),\mathbf{V}(t); f\big)\circ dW^\nu(t),
\end{align}
\end{subequations}

\noindent
on the time interval $[0,T]$, if and only if they satisfy the following variational principle

\begin{equation}
\label{eq: stochastic Lagrange-d'Alembert principle for Vlasov-Poisson}
\delta_{(\mathbf{X},\mathbf{V},\mathbf{P})} \hat S[\mathbf{X},\mathbf{X},\mathbf{V},\mathbf{P}] + m\int_0^T \mathbf{G}(\mathbf{X},\mathbf{V}; f)\cdot \delta \mathbf{X} \, dt + m\sum_{\nu=1}^M \int_0^T\mathbf{g}_\nu (\mathbf{X},\mathbf{V}; f)\cdot \delta \mathbf{X}\circ dW^\nu(t)  = 0
\end{equation}

\noindent
for arbitrary variations $\delta\mathbf{X}\in C^1_{\Omega,T}$, and $\delta\mathbf{V}, \delta\mathbf{P}\in C^0_{\Omega,T}$, with  $\delta\mathbf{X}(0)=\delta\mathbf{X}(T)=0$ almost surely, where the action functional $\hat S$ is given by \eqref{eq: Action functional for Vlasov-Poisson - formula}.
\end{thm}

\begin{proof}
Analogous to the proof of Theorem~\ref{thm: Stochastic Lagrange-d'Alembert principle for particles}.\\
\end{proof}

\paragraph{Remark.} It is straightforward to see that Equations \eqref{eq: SDEs for the Lagrange-d'Alembert principle for Vlasov-Poisson}, together with \eqref{eq: Electrostatic potential - solution} and \eqref{eq: Electric field - solution}, are equivalent to the system of equations \eqref{eq: Charge and current densities in terms of expected values 2}, \eqref{eq: Poisson's equation}, and \eqref{eq: SDE for X and V for Vlasov-Poisson}. The Lagrange-d'Alembert principle \eqref{eq: stochastic Lagrange-d'Alembert principle for Vlasov-Poisson} is unusual in that the variations of the action functional $\hat S$ with respect to the argument $\mathbf{Y}$ are omitted. Thanks to such a form, however, the action functional does not require an expected value, and the collisional effects can be correctly included. A similar idea to solve Poisson's equation and plug the solution into the action functional was presented in \cite{YeMorrisonActionPrinciples}, where the authors proposed a variational principle for the collisionless Vlasov-Poisson equations. In that approach the energy of the electric field was also included in the variational principle, and the variations were taken with respect to all arguments of the action functional. This approach could be adapted to the stochastic reformulation of the Vlasov-Poisson equations, but the corresponding action functional would have a form similar to \eqref{eq: Action functional - formula}, that is, it would need to contain an expected value, and therefore we would face a similar difficulty as for the Vlasov-Maxwell equations in Section~\ref{sec: The stochastic Lagrange-d'Alembert principle for Vlasov-Maxwell}.

%%%%%%%%%%%%%%%%%%%%%%%%%%%%%%%%%%%%%%%%%%%%%%%%%%%%%%%%%%%%%%%%%%%%%%%%%%%%%%%%%%%%
%  SUMMARY
%%%%%%%%%%%%%%%%%%%%%%%%%%%%%%%%%%%%%%%%%%%%%%%%%%%%%%%%%%%%%%%%%%%%%%%%%%%%%%%%%%%%
\section{Summary and future work}
\label{sec:Summary}

In this work we have considered novel stochastic formulations of the collisional Vlasov-Maxwell and Vlasov-Poisson equations, and we have identified new stochastic variational principles underlying these formulations. We have also proposed a stochastic particle method for the Vlasov-Maxwell equations and proved the corresponding stochastic variational principle.

Our work can be extended in several ways. The stochastic variational principle introduced in Section~\ref{sec: Variational principle} can be used to construct stochastic variational particle-in-cell numerical algorithms for the collisional Vlasov-Maxwell and Vlasov-Poisson equations. Variational integrators are an important class of geometric integrators. This type of numerical schemes is based on discrete variational principles and provides a natural framework for the discretization of Lagrangian systems, including forced, dissipative, or constrained ones. These methods have the advantage that they are symplectic when applied to systems without forcing, and in the presence of a symmetry, they satisfy a discrete version of Noether's theorem. For this reason they demonstrate superior performance in long-time simulations; see \cite{HallLeokSpectral}, \cite{JaySPARK}, \cite{KaneMarsden2000}, \cite{LeokShingel}, \cite{LeokZhang}, \cite{MarsdenWestVarInt}, \cite{OberBlobaum2016}, \cite{OberBlobaum2015}, \cite{RowleyMarsden}, \cite{TyranowskiDesbrunLinearLagrangians}, \cite{VankerschaverLeok}. Variational integrators were introduced in the context of finite-dimensional mechanical systems, but were later generalized to Lagrangian field theories (see \cite{MarsdenPatrickShkoller}) and applied in many computations, for example in elasticity, electrodynamics, fluid dynamics, or plasma physics; see \cite{KrausPHD}, \cite{LewAVI}, \cite{Pavlov}, \cite{Squire2012}, \cite{SternDesbrun}, \cite{TyranowskiDesbrunRAMVI}, \cite{XiaoLiuQin2013}, \cite{XiaoQinLiu2018}. Stochastic variational integrators were first introduced in \cite{BouRabeeSVI} and further studied in \cite{BouRabeeConstrainedSVI}, \cite{HolmTyranowskiSolitons}, \cite{HolmTyranowskiGalerkin}, \cite{KrausTyranowski2019}, \cite{WangPHD}.

 In Section~\ref{sec: Variational principle for the Vlasov-Maxwell equation} we have proposed a general action functional for the collisional Vlasov-Maxwell equations. However, we have also determined that in order to prove a relevant variational principle, the domain of this action functional has to be restricted in a way compatible with the collision operator of interest. We have shown that for a class of collision operators with constant diffusion terms, a suitable subdomain is an affine subspace (i.e., a submanifold). A natural continuation of our work would be to investigate submanifolds of \eqref{eq: Action functional - full domain} which are suitable for other collision operators.

Another aspect worth a more detailed investigation is the issue of existence and uniqueness of the solutions of the stochastic reformulations presented in this work, which are nontrivial systems of coupled stochastic and partial differential equations. This question is closely connected to the issue of existence and uniqueness of the solutions of the original collisional Vlasov-Maxwell and the Vlasov-Poisson equations. General results are available in the collisionless case (see, e.g., \cite{DegondNeunzert1986}, \cite{Wollman1984}, \cite{Wollman1987}), but the theory for the collisional equations is less developed (see \cite{Degond1986}, \cite{DuanStrain2011}, \cite{NeunzertVlasovFokkerPlanck1984}, \cite{Ono2001}, \cite{Strain2006}, \cite{WangVlasovMaxwell2019}, \cite{WangFokkerPlanck2020} and the references therein).

Furthermore, our stochastic Lagrange-d'Alembert approach could also be adapted to relativistic plasmas (see \cite{BraamsKarney1989}), and to variational principles with phase-space Lagrangians appearing in gyrokinetic (\cite{BottinoSonnendrucker2015}, \cite{Brizard2000}, \cite{SugamaGyrokinetic2000}) and guiding-center theories (\cite{BrizardTronci2016}, \cite{CaryBrizard2009}, \cite{Pfirsch1984}, \cite{PfirschMorrison1985}). In particular, considering stochastic extensions of the variational principles proposed in \cite{BrizardTronci2016} could offer an alternative stochastic description of anomalous transport in magnetically confined plasmas (see \cite{Balescu1994}, \cite{Eijnden1998}).

Finally, as is typical for particle methods in general, the stochastic particle discretization proposed in Section~\ref{sec:Particle discretization} will require a large number of particles for accurate numerical simulations, which is computationally expensive. Structure-preserving model reduction methods (see \cite{AfkhamHesthaven2017}, \cite{PengMohseni2016}) have been recently successfully applied to particle discretizations of the collisionless Vlasov equation (see \cite{TyranowskiKraus2021}). It would be of great practical interest to combine our results with model reduction techniques in order to develop new efficient structure-preserving data-driven numerical methods for the collisional Vlasov-Maxwell equations.

%%%%%%%%%%%%%%%%%%%%%%%%%%%%%%%%%%%%%%%%%%%%%%%%%%%%%%%%%%%%%%%%%%%%%%%%%%%%%%%%%%%%
%  ACKNOWLEDGMENTS
%%%%%%%%%%%%%%%%%%%%%%%%%%%%%%%%%%%%%%%%%%%%%%%%%%%%%%%%%%%%%%%%%%%%%%%%%%%%%%%%%%%%
\section*{Acknowledgements}

We would like to thank Christopher Albert, Darryl Holm, Michael Kraus, Omar Maj, Houman Owhadi, Eric Sonnendr\"{u}cker, and Cesare Tronci for useful comments and references. The study is a contribution to the Reduced Complexity Models grant number ZT-I-0010 funded by the Helmholtz Association of German Research Centers.

%%%%%%%%%%%%%%%%%%%%%%%%%%%%%%%%%%%%%%%%%%%%%%%%%%%%%%%%%%%%%%%%%%%%%%%%%%%%%%%%%%%%
%  SECTIONS REQUIRED BY PROCEEDINGS OF THE ROYAL SOCIETY - COMMENT OUT FOR ARXIV
%%%%%%%%%%%%%%%%%%%%%%%%%%%%%%%%%%%%%%%%%%%%%%%%%%%%%%%%%%%%%%%%%%%%%%%%%%%%%%%%%%%%
%\subsection*{Data accessibility}
%This paper has no data.
%\subsection*{Competing interests}
%We have no competing interests.
%\subsection*{Authors' contributions}
%Does not apply.
%\subsection*{Acknowledgements}
%We would like to thank Christopher Albert, Darryl Holm, Michael Kraus, Omar Maj, Houman Owhadi, Eric Sonnendr\"{u}cker, and Cesare Tronci for useful comments and references.
%\subsection*{Funding statement}
%The study is a contribution to the Reduced Complexity Models grant number ZT-I-0010 funded by the Helmholtz Association of German Research Centers.
%\subsection*{Ethics statement}
%Does not apply.

%%%%%%%%%%%%%%%%%%%%%%%%%%%%%%%%%%%%%%%%%%%%%%%%%%%%%%%%%%%%%%%%%%%%%%%%%%%%%%%%%%%%
%  APPENDICES
%%%%%%%%%%%%%%%%%%%%%%%%%%%%%%%%%%%%%%%%%%%%%%%%%%%%%%%%%%%%%%%%%%%%%%%%%%%%%%%%%%%%
\appendix

\section{The variations of the action functional $S$}
\label{sec: The variations of the action functional S}
We will define the variation of $S$ with respect to the variation $\delta \mathbf{X}_a \in C^1_{\Omega,T}$ of the argument $\mathbf{X}_a$ as

\begin{equation}
\label{eq: Variation of S wrt X - definition}
\delta_{\mathbf{X}_a}S = \frac{d}{d\epsilon}\bigg|_{\epsilon=0}S[\mathbf{X}_1,\ldots,\mathbf{X}_a+\epsilon \delta\mathbf{X}_a,\ldots,\mathbf{X}_N,\mathbf{V}_1,\ldots,\mathbf{V}_N, \mathbf{P}_1,\ldots,\mathbf{P}_N, \varphi, \mathbf{A}].
\end{equation}

\noindent
Since the potentials $\varphi$ and $\mathbf{A}$ are $C^2$, and the processes $\mathbf{X}_b$, $\mathbf{V}_b$, and $\mathbf{P}_b$ are almost surely continuous, we can use a dominated convergence argument to interchange the differentiation with respect to $\epsilon$ and integration with respect to $t$ to obtain

\begin{equation}
\label{eq: Variation of S wrt X - calculation 1}
\delta_{\mathbf{X}_a}S = \frac{N_{tot}}{N} \int_0^T \bigg( -q \nabla_x \varphi(\mathbf{X}_a,t)\cdot \delta\mathbf{X}_a+q\sum_{i,j=1}^3V^j\frac{\partial A^j}{\partial x^i}(\mathbf{X}_a,t)\delta X_a^i+\mathbf{P}_a\cdot\delta \mathbf{\dot X}_a\bigg)\,dt.
\end{equation}

\noindent
Since $\delta\mathbf{X}_a$ is almost surely differentiable, we have that its stochastic differential is simply $d\delta\mathbf{X}_a=\delta \mathbf{\dot X}_a\,dt$. Furthermore, both $\delta \mathbf{X}_a$ and $\mathbf{P}_a$ are almost surely continuous semimartingales, therefore using the integration by parts formula for semimartingales (see \cite{ProtterStochastic}) we can write

\begin{equation}
\label{eq: Integration by parts}
\int_0^T \mathbf{P}_a\cdot\delta \mathbf{\dot X}_a\,dt = \int_0^T \mathbf{P}_a \circ d\delta \mathbf{X}_a=\mathbf{P}_a(t) \cdot \delta \mathbf{X}_a(t) \Big|_0^T - \int_0^T \delta \mathbf{X}_a\circ d\mathbf{P}_a,
\end{equation}

\noindent
where the Stratonovich integrals are understood in the sense that $\int \delta \mathbf{X}_a\circ d\mathbf{P}_a = \sum_i \int \delta X^i_a\circ dP^i_a$. By substituting \eqref{eq: Integration by parts} in \eqref{eq: Variation of S wrt X - calculation 1}, we obtain \eqref{eq: Variation of S wrt X - calculation 2}. Variations with respect to $\delta \mathbf{V}_a, \delta \mathbf{P}_a \in C^0_{\Omega,T}$ are defined analogously to \eqref{eq: Variation of S wrt X - definition}. Similar computations (note that integration by parts is not necessary) yield \eqref{eq: Variation of S wrt V} and \eqref{eq: Variation of S wrt P}, respectively.

The variation of $S$ with respect to the variation $\delta \mathbf{A} \in \mathfrak{X}_0(\mathbb{R}^3)$ of the vector potential $\mathbf{A}$ is defined as

\begin{equation}
\label{eq: Variation of S wrt A - definition}
\delta_{\mathbf{A}}S = \frac{d}{d\epsilon}\bigg|_{\epsilon=0}S[\mathbf{X}_1,\ldots,\mathbf{X}_N,\mathbf{V}_1,\ldots,\mathbf{V}_N, \mathbf{P}_1,\ldots,\mathbf{P}_N, \varphi, \mathbf{A}+\epsilon \delta\mathbf{A}].
\end{equation}

\noindent
Switching the order of differentiation and integration, integrating by parts, and using the fact that $\delta \mathbf{A}$ is compactly supported, one arrives at \eqref{eq: Variation of S wrt A - calculation}, where in the derivations we have used \eqref{eq: Charge and current densities from the law of large numbers 3} and

\begin{align}
\label{eq: Integral of the current density}
\frac{N_{tot}}{N}\sum_{b=1}^N\bigg[\int_0^T q\mathbf{V}_b(t)\cdot\delta\mathbf{A}(\mathbf{X}_b,t)\,dt\bigg] &= \int_0^T\!\!\!\!\! \int_{\mathbb{R}^3} \frac{qN_{tot}}{N}\sum_{b=1}^N\big[q\mathbf{V}_b(t)\delta(\mathbf{x}-\mathbf{X}_b(t))\big]\cdot\delta\mathbf{A}(\mathbf{x},t)\,d^3\mathbf{x}dt \nonumber\\
&=\int_0^T\!\!\!\!\! \int_{\mathbb{R}^3}\mathbf{J}(\mathbf{x},t)\cdot\delta\mathbf{A}(\mathbf{x},t)\,d^3\mathbf{x}dt,
\end{align}

\noindent
and the remaining calculations are standard, and can be found in, e.g., \cite{Evans}, \cite{JacksonElectrodynamics}. The variation of $S$ with respect to the variation $\delta \varphi \in \mathfrak{X}_0(\mathbb{R})$ of the scalar potential $\varphi$ is defined in a similar fashion, and after similar calculations one obtains \eqref{eq: Variation of S wrt phi - calculation}.

\section{Proof of Lemma~\ref{thm: Fundamental lemma of calculus of variations}}
\label{sec: Proof of Lemma}

\begin{proof}
Suppose that \eqref{eq: Lemma2 equation 2} holds. Then \eqref{eq: Lemma2 equation 1} follows from the associativity property of the Stratonovich integral (see, e.g., the proof of Theorem~2.1 in \cite{HolmTyranowskiGalerkin}). Conversely, assume that \eqref{eq: Lemma2 equation 1} is satisfied, and let us prove that \eqref{eq: Lemma2 equation 2} follows. Our reasoning very closely follows the proof of Theorem~3.3 in \cite{BouRabeeSVI}. Pick any time $t\in[0,T]$. We will use $\mathbf{e}_1$, $\mathbf{e}_2$, and $\mathbf{e}_3$ to denote the standard Cartesian basis vectors for $\mathbb{R}^3$. Pick a basis vector $\mathbf{e}_i$. The condition \eqref{eq: Lemma2 equation 1} in particular holds for $\mathbf{Z}$'s which are $C^1$ functions of time, i.e., non-random. The main idea of the proof is to construct a one-parameter family of $C^1$ functions $\mathbf{Z}_\epsilon$ which converge to $\mathbbm{1}_{[0,t]}\mathbf{e}_i$ as $\epsilon \longrightarrow 0$, and show that the integral in \eqref{eq: Lemma2 equation 1} converges almost surely to the integral in \eqref{eq: Lemma2 equation 2}. Let us introduce the notation

\begin{align}
\label{eq: Definition of I}
I(\mathbf{X}, \mathbf{V}, \mathbf{P}, \mathbf{Z}) &= \int_0^T \Big( \mathbf{Z}(\tau)\circ d\mathbf{P}(\tau)-\mathbf{R}(\mathbf{X},\mathbf{V})\cdot\mathbf{Z}(\tau)\,d\tau - \sum_{\nu=1}^M\mathbf{r}_\nu(\mathbf{X},\mathbf{V})\cdot\mathbf{Z}(\tau)\circ dW^\nu(\tau) \Big),\\
\label{eq: Definition of I*}
I^*(\mathbf{X}, \mathbf{V}, \mathbf{P})&= \int_0^T \Big( \mathbbm{1}_{[0,t]}\mathbf{e}_i\circ d\mathbf{P}(\tau)-\mathbf{R}(\mathbf{X},\mathbf{V})\cdot\mathbbm{1}_{[0,t]}\mathbf{e}_i\,d\tau - \sum_{\nu=1}^M\mathbf{r}_\nu(\mathbf{X},\mathbf{V})\cdot\mathbbm{1}_{[0,t]}\mathbf{e}_i\circ dW^\nu(\tau) \Big) \nonumber \\
                                       &=  \int_0^t \Big(dP^i(\tau)-R^i(\mathbf{X}(\tau),\mathbf{V}(\tau))\,d\tau - \sum_{\nu=1}^M r^i_\nu(\mathbf{X}(\tau),\mathbf{V}(\tau))\circ dW^\nu(\tau) \Big).
\end{align}

\noindent
Define the functions $h_1:[0,\epsilon]\longrightarrow [0,1]$ and $h_2:[t-\epsilon,t]\longrightarrow [0,1]$ by the formulas

\begin{equation}
\label{eq: Definition of the functions h}
h_1(\tau) = 2 \frac{\tau}{\epsilon}-\frac{\tau^2}{\epsilon^2}, \qquad\qquad h_2(\tau)= 
\begin{cases}
-\frac{2}{\epsilon^2}(\tau-t+\epsilon)^2+1                              & \text{if $t-\epsilon \leq \tau \leq t-\frac{\epsilon}{2}$,} \\
\frac{2}{\epsilon^2}(\tau-t+\epsilon)^2 - \frac{4}{\epsilon} (\tau-t+\epsilon) + 2   & \text{if $t-\frac{\epsilon}{2} < \tau \leq t$.}
\end{cases}
\end{equation}

\noindent
Note that $h_1(0)=h_2(t)=0$, $h_1(\epsilon)=h_2(t-\epsilon)=1$, and $h'_1(\epsilon)=h'_2(t-\epsilon)=h'_2(t)=0$. Define further the family of functions $\mathbf{Z}_\epsilon$ by the formula

\begin{equation}
\label{eq: Sequence of differentiable Z functions}
\mathbf{Z}_\epsilon(\tau) =
  \begin{cases} 
      h_1(\tau) \mathbf{e}_i      & \text{if $0 \leq \tau \leq \epsilon$,} \\
      \mathbf{e}_i              & \text{if $\epsilon < \tau < t-\epsilon$,} \\
      h_2(\tau) \mathbf{e}_i    & \text{if $t-\epsilon \leq \tau \leq t$,} \\
			0                         & \text{if $t < \tau \leq T$.}
   \end{cases}
\end{equation} 

\noindent
It is easy to see that $\mathbf{Z}_\epsilon$ is continuously differentiable\footnote{Note that our definition \eqref{eq: Sequence of differentiable Z functions} is slightly different from the corresponding definition in \cite{BouRabeeSVI}, because the test functions used in \cite{BouRabeeSVI} are in fact not differentiable at $\tau=t$. This, however, is of little consequence for the rest of the proof.} on $[0,T]$, and converges to $\mathbbm{1}_{[0,t]}\mathbf{e}_i$ in the $L^2$ norm as $\epsilon$ goes to zero. Using \eqref{eq: Definition of I}, \eqref{eq: Definition of I*}, \eqref{eq: Definition of the functions h}, and \eqref{eq: Sequence of differentiable Z functions}, we have

\begin{align}
\label{eq: Calculating I*-I}
I^*(\mathbf{X}, &\mathbf{V},\mathbf{P})-I(\mathbf{X}, \mathbf{V}, \mathbf{P}, \mathbf{Z}_\epsilon)= \nonumber \\
&\int_0^\epsilon \bigg( (1-h_1(\tau))\circ dP^i(\tau) - (1-h_1(\tau))R^i(\mathbf{X}, \mathbf{V})\,d\tau-\sum_{\nu=1}^M (1-h_1(\tau))r^i_\nu(\mathbf{X}, \mathbf{V})\circ dW^\nu \bigg) \nonumber \\
+&\int_{t-\epsilon}^t \bigg( (1-h_2(\tau))\circ dP^i(\tau) - (1-h_2(\tau))R^i(\mathbf{X}, \mathbf{V})\,d\tau -\sum_{\nu=1}^M (1-h_2(\tau))r^i_\nu(\mathbf{X}, \mathbf{V})\circ dW^\nu \bigg).
\end{align}

\noindent
By definition, the Stratonovich integrals in \eqref{eq: Calculating I*-I} can be expressed in terms of the It\^{o} integrals as

\begin{align}
\label{eq: Stratonovich integrals in terms of Ito integrals}
\int_0^\epsilon (1-h_1(\tau))r^i_\nu(\mathbf{X}, \mathbf{V})\circ dW^\nu &= \int_0^\epsilon (1-h_1(\tau))r^i_\nu(\mathbf{X}, \mathbf{V})\, dW^\nu + \frac{1}{2}\Big[(1-h_1(\tau))r^i_\nu(\mathbf{X}, \mathbf{V}),W^\nu(\tau)\Big]^\epsilon_0, \nonumber \\
\int_{t-\epsilon}^t(1-h_2(\tau))r^i_\nu(\mathbf{X}, \mathbf{V})\circ dW^\nu &= \int_{t-\epsilon}^t (1-h_2(\tau))r^i_\nu(\mathbf{X}, \mathbf{V})\, dW^\nu + \frac{1}{2}\Big[(1-h_2(\tau))r^i_\nu(\mathbf{X}, \mathbf{V}),W^\nu(\tau)\Big]_{t-\epsilon}^t,
\end{align}

\noindent
for each $\nu=1,\ldots,M$, where $[\cdot,\cdot]$ denotes the quadratic covariation process. Since the quadratic covariation of almost surely continuous semimartingales is itself a semimartingale with almost surely continuous paths (see Theorem~23 in Chapter~II.6 of \cite{ProtterStochastic}), we have that

\begin{equation}
\label{eq: Quadratic covariation limit 1}
\Big[(1-h_1(\tau))r^i_\nu(\mathbf{X}(\tau), \mathbf{V}(\tau)),W^\nu(\tau)\Big]^\epsilon_0 \longrightarrow (1-h_1(0))r^i_\nu(\mathbf{X}(0), \mathbf{V}(0))W^\nu(0) = 0 \qquad \text{a.s. as $\epsilon \longrightarrow 0$},
\end{equation}

\noindent
since $W^\nu(0) = 0$ almost surely. In a similar fashion we show 

\begin{equation}
\label{eq: Quadratic covariation limit 2}
\Big[(1-h_2(\tau))r^i_\nu(\mathbf{X}(\tau), \mathbf{V}(\tau)),W^\nu(\tau)\Big]_{t-\epsilon}^t \longrightarrow 0 \qquad \text{a.s. as $\epsilon \longrightarrow 0$}.
\end{equation}

\noindent
Using \eqref{eq: Calculating I*-I} and \eqref{eq: Stratonovich integrals in terms of Ito integrals}, we have the estimate

\begin{align}
\label{eq: Estimating I*-I}
|I^*(\mathbf{X}, &\mathbf{V},\mathbf{P})-I(\mathbf{X}, \mathbf{V}, \mathbf{P}, \mathbf{Z}_\epsilon)| \leq \nonumber \\
&\underbrace{\bigg|\int_0^\epsilon \bigg( (1-h_1(\tau))\circ dP^i(\tau) - (1-h_1(\tau))R^i(\mathbf{X}, \mathbf{V})\,d\tau-\sum_{\nu=1}^M (1-h_1(\tau))r^i_\nu(\mathbf{X}, \mathbf{V})\,dW^\nu \bigg) \bigg| }_{\Gamma_1}\nonumber \\
+&\underbrace{\bigg|\int_{t-\epsilon}^t \bigg( (1-h_2(\tau))\circ dP^i(\tau) - (1-h_2(\tau))R^i(\mathbf{X}, \mathbf{V})\,d\tau -\sum_{\nu=1}^M (1-h_2(\tau))r^i_\nu(\mathbf{X}, \mathbf{V})\, dW^\nu \bigg) \bigg| }_{\Gamma_2}\nonumber \\
+&\frac{1}{2}\sum_{\nu=1}^M \bigg| \Big[(1-h_1(\tau))r^i_\nu(\mathbf{X}, \mathbf{V}),W^\nu(\tau)\Big]^\epsilon_0 \bigg| + \frac{1}{2} \sum_{\nu=1}^M \bigg|\Big[(1-h_2(\tau))r^i_\nu(\mathbf{X}, \mathbf{V}),W^\nu(\tau)\Big]_{t-\epsilon}^t \bigg|.
\end{align}

\noindent
By bounding the integrands and using the It\^{o} isometry theorem, it is shown in \cite{BouRabeeSVI} that $\Gamma_1 \longrightarrow 0$ and $\Gamma_2 \longrightarrow 0$ in mean-square as $\epsilon \longrightarrow 0$, and consequently, by invoking the Borel-Cantelli lemma, there exists a subsequence $(\epsilon_n)$ such that $\epsilon_n \longrightarrow 0$ as $n\longrightarrow \infty$, for which $\Gamma_1 \longrightarrow 0$ and $\Gamma_2 \longrightarrow 0$ almost surely. Together with \eqref{eq: Quadratic covariation limit 1} and \eqref{eq: Quadratic covariation limit 2}, this means that $I(\mathbf{X}, \mathbf{V}, \mathbf{P}, \mathbf{Z}_{\epsilon_n}) \longrightarrow I^*(\mathbf{X}, \mathbf{V},\mathbf{P})$ almost surely. Given the assumption \eqref{eq: Lemma2 equation 1}, we have that $I^*(\mathbf{X}, \mathbf{V},\mathbf{P}) = 0$ almost surely, which completes the proof.

\end{proof}

%%%%%%%%%%%%%%%%%%%%%%%%%%%%%%%%%%%%%%%%%%%%%%%%%%%%%%%%%%%%%%%%%%%%%%%%%%%%%%%%%%%%
%  BIBLIOGRAPHY
%%%%%%%%%%%%%%%%%%%%%%%%%%%%%%%%%%%%%%%%%%%%%%%%%%%%%%%%%%%%%%%%%%%%%%%%%%%%%%%%%%%%
%\bibliographystyle{abbrv}
%\bibliographystyle{plain}
%\bibliography{bibliography}

\end{document}